\documentclass[twoside]{article}
\usepackage[utf8]{inputenc}
\usepackage[T1]{fontenc}
\usepackage{microtype}
\usepackage{graphicx}
\usepackage[round]{natbib}
\usepackage{bbm}
\usepackage{booktabs} 
\usepackage{apxproof}
\usepackage{bbold}

\usepackage{amsmath,amssymb,amsthm}
\usepackage{amsthm}
\usepackage{amsmath,amsfonts,bm}
\makeatletter
\newtheorem*{rep@theorem}{\rep@title}
\newcommand{\newreptheorem}[2]{%
\newenvironment{rep#1}[1]{%
 \def\rep@title{#2 \ref{##1}}%
 \begin{rep@theorem}}%
 {\end{rep@theorem}}}
\makeatother

\newtheorem{definition}{Definition}[section]
\newtheoremrep{theorem}{Theorem}
\newtheorem{example}{Example}
\newtheoremrep{proposition}{Proposition}
\newtheorem{lemma}{Lemma}

\newcommand{\RNum}[1]{\uppercase\expandafter{\romannumeral #1\relax}}

\newcommand{\ours}{disc decomposition\xspace}

\usepackage{mathtools}

\newcommand\norm[1]{\left\lVert#1\right\rVert} 

\DeclareMathOperator*{\argmin}{\arg\!\min}

\DeclareMathOperator*{\Int}{Int}
\DeclareMathOperator*{\hull}{hull}

\newcommand{\cut}[1]{}

\newcommand{\removelatexerror}{\let\@latex@error\@gobble}

\newcommand{\vect}[1]{\mathbf{#1}}
\newcommand{\mat}[1]{\mathbf{#1}}

\def\eqref#1{Eq.~\ref{#1}}

\def\floor#1{\lfloor #1 \rfloor}
\def\1{\bm{1}}

\def\rvu{{\mathbf{i}}}

\def\rvu{{\mathbf{u}}}

\def\rmA{{\mathbf{A}}}

\def\rmP{{\mathbf{P}}}

\DeclareMathAlphabet{\mathsfit}{\encodingdefault}{\sfdefault}{m}{sl}
\SetMathAlphabet{\mathsfit}{bold}{\encodingdefault}{\sfdefault}{bx}{n}

\def\sN{{\mathbb{N}}}

\def\sP{{\mathbb{P}}}

\def\sR{{\mathbb{R}}}

\DeclareMathOperator{\logit}{logit}

\usepackage[pagebackref=true]{hyperref}
\renewcommand*{\backrefalt}[4]{%
    \ifcase #1 \footnotesize{(Not cited.)}%
    \or        \footnotesize{(Cited on page~#2)}%
    \else      \footnotesize{(Cited on pages~#2)}%
    \fi}

\usepackage{subcaption}
\usepackage[inline, shortlabels]{enumitem}
\usepackage{bbm}

\usepackage{caption}

\usepackage{makecell}
\usepackage{colortbl}
\usepackage{xcolor}

\newcommand{\rebutal}[1]{\textcolor{black}{#1}}

\definecolor{linkcolor}{RGB}{83,83,182}
\definecolor{citecolor}{RGB}{128,0,128}
\hypersetup{
    colorlinks=true,
    citecolor=citecolor,
    linkcolor=linkcolor,
    urlcolor=linkcolor
}
\newcommand{\normin}[1]{ \lVert {#1} \rVert}
\newcommand{\payoffprob}[2]{ \mathrm{\mathrm{P}}_{#1#2}}
\newcommand{\payoffprobhat}[2]{ \mathrm{\mathrm{\hat P}}_{#1#2}}
\newcommand{\payoffprobmat}{ \mathrm{\mathbf{P}} }
\newcommand{\payoffanti}[2]{ \mathrm{\mathrm{A}}_{#1#2}}
\newcommand{\player}[1]{ \mathrm{u}_{#1}}
\newcommand{\playertilde}[1]{ \tilde{\mathrm{u}}_{#1}}
\newcommand{\playertildev}[1]{ \tilde{\mathrm{v}}_{#1}}
\newcommand{\playerv}[1]{ \mathrm{v}_{#1}}

\usepackage{siunitx}
\usepackage{algorithm}
\usepackage{algorithmic}
\usepackage[titlenumbered,ruled,noend,algo2e]{algorithm2e}

\usepackage[nameinlink]{cleveref}

\usepackage{aliascnt}
\usepackage{thmtools,thm-restate}

\makeatletter

\def\endproblem{\eqno \hbox{\@eqnnum}$$\@ignoretrue}
\makeatother
\crefname{model}{Model}{Models}
\Crefname{problem}{Problem}{Problems}
\crefname{problem}{Pb.}{Pbs.}
\crefname{algorithm}{Algorithm}{Algorithms}
\crefname{figure}{Figure}{Figures}
\crefname{proposition}{Proposition}{Propositions}
\crefname{appendix}{Appendix}{Appendix}
\crefname{assumption}{Assumption}{Assumptions}

\usepackage[accepted]{aistats2023}

\begin{document}
%

%

\twocolumn[

\aistatstitle{On the Limitations of the Elo, Real-World Games are Transitive, not Additive}

\aistatsauthor{ Quentin Bertrand \And Wojciech Marian Czarnecki \And  Gauthier Gidel }

\aistatsaddress{ Mila and Université de Montréal \And  VoyLab \And Mila and Université de Montréal\\ Canada CIFAR AI Chair } ]

\begin{abstract}
The Elo score has been extensively used to rank players by their skill or strength in competitive games such as chess, go, or StarCraft II. The Elo score implicitly assumes games have a strong additive---hence transitive---component.
In this paper, we investigate the challenge of identifying transitive components in games. As a starting point, we show that the Elo score provably fails to extract the transitive component of some elementary transitive games.
Based on this observation, we propose an alternative ranking system that properly extracts the transitive components in these games. Finally, we conduct an in-depth empirical validation on real-world game payoff matrices: it shows significant prediction performance improvements compared to the Elo score.

\end{abstract}
\section{INTRODUCTION}
Accurately evaluating and ranking players in games has been a concern for decades. For instance, in the context of chess,~\citet{elo1961uscf} stated that ``The rating system shall provide: as close an estimate as possible of the current playing strength of an individual as computed from his performance in competition with other players''.
The goal of the Elo score was to rank the players by their ``skill'' or ``strength''.
In this model, the probability of a confrontation outcome between two players is an increasing function of the difference in their Elo scores.

\looseness=-1
More recently, a surge of interest in evaluation and ranking for games has been triggered by the successes of multi-agent reinforcement learning agents in increasingly complex environments~\citep{vinyals2019grandmaster,Berner2019,liu2021motor}.
Thus, having accurate and scalable ways to evaluate agents becomes crucial to monitoring agents' learning, selecting the best agents, and performing evolutionary algorithms \citep{Jaderberg2019}.
This interest yields alternatives and extensions of the Elo score~\citep{herbrich2006trueskill,minka2018trueskill} which take into account the sequential aspect of the games~\citep{Cardoso2019}, to quantify potential uncertainty on the Elo score~\citep{Glickman2001,herbrich2006trueskill,Sismanis2010}.
In addition, several regularizations can be added to the Elo score to improve its generalization~\citep{Sismanis2010} and sample complexity~\citep{yan2022learning}.

Self-play learning algorithms can now easily beat humans at
games such as chess, shogi, or go \citep{Silver2018,Bansal2018,Jaderberg2019},
which are considered mostly transitive (if $a$ beats $b$, and $b$ beats $c$, then $a$ beats $c$).
However, designing learning algorithms for highly non-transitive games such as StarCraft~II\footnote{StarCraft II is considered to have cyclic components: $a$ beats $b$, who beats $c$, who beats $a$.} is much more challenging~\citep{Lanctot2017,Balduzzi2018,Balduzzi2019}.
For instance,~\citet{vinyals2019grandmaster} explains that self-play ``is more
forgetful'' when trained to play StarCraft II.
To ``avoid cycles'',~\citet{vinyals2019grandmaster} considered a form of play against a population of agents ``by playing against all previous players in the league''.

Even though cyclic behaviors are observed in practice: in the games Age of Empires II \citep{Horrigan2022}, StarCraft II, \citep{Vinyals2017} or Dota 2 \citep{Berner2019}.
It seems that for many players, the transitive property is still valid (see for instance \Cref{fig:real_chess_data} in \Cref{sec:experiments}).
That is why the Elo score has been used to rank participants for multiple games and sports, from chess to
basketball \citep{NBAElo}
or
football \citep{FIFAElo}.

These notions of cyclicity and transitivity remain relatively qualitative and usually require expert domain knowledge to be assessed.
To understand why some games are more cyclic than others, we propose quantifying the amount of cyclicity and transitivity in real-world games.
More precisely, given $n$ players
and their empirical payoff matrix $(\payoffprob{i}{j})_{1\leq i,j \leq n}$,\footnote{We would like to emphasize that our goal differs from standard empirical game theory~\citep{walsh2003choosing,wellman2006methods} which aims at approximating the Nash of the game.}
we propose new scores to answer the question:
\begin{center}
    \begin{minipage}{.95\linewidth}
        \centering
        \emph{can one quantify the amount of transitivity
        of a \rebutal{zero-sum two-player game} from empirical data?}
    \end{minipage}
\end{center}
%
%
The contributions of this work are the following:
\begin{itemize}[
    noitemsep,itemjoin = \quad,topsep=0pt,parsep=0pt,partopsep=0pt, leftmargin=*]
    \item We propose a \ours (\Cref{thm:at_most_one_transitive}) which yields a quantitative and tractable definition of the amount of transitivity of a real-world game.
    \item We compute this quantity for several real-world games, including chess and StarCraft II, which yields better predictions for new matchup outcomes.
\end{itemize}
The paper is organized as follows.
\Cref{sec:background} provides some background on zero-sum games and the Elo score.
The proposed \ours, based on the normal decomposition of skew-symmetric matrices, is developed in \Cref{sec:proposed_approach}.
Experiments on payoff matrices coming from real-world games played by bots and humans are provided in \Cref{sec:experiments}.
All proofs can be found in the appendix.
\paragraph{Notation}
Capital bold letters denote matrices, and lowercase bold letters denote vectors.
The set of integers from $1$ to $n$ is denoted $[n]$.
The sigmoid function is written as $\sigma$, and its inverse is written as $\logit$.
The binary cross entropy
$(x, \hat x) \mapsto
- x \log (\hat x)
- (1 - x) \log (1 - \hat x)$ is denoted $\mathrm{bce}$.
\paragraph{Related Work}
\looseness=-1
We consider~\citet{Balduzzi2018,Balduzzi2019} to be our closest related works.
Relying on combinatorial Hodge theory \citep{Jiang2011},
\citet{Balduzzi2018} proposed an extension of the Elo score, called $m$-Elo, which can express potential cyclic components.
Given a payoff matrix $\mat{P}$ (\Cref{def:empirical_game}) they impose a transitive component
$\rvu^{\mathrm{Trans}}
\triangleq
(\sum_j \logit(\mat{P})_{i j})_{i \in [n]}$,
then they compute the normal decomposition (\Cref{thm:antisym_pca}) of the matrix $\logit(\mat{P}) - (\rvu^{\mathrm{Trans}} \mathbf{1}^\top - \mathbf{1} \rvu^{\mathrm{Trans} \top})$.
Unlike the $m$-Elo, we do not impose the transitive component to correspond to an Elo score.
In this work, we instead propose to directly compute the normal decomposition of $\logit(\mat{P})$ and provide a principled result (\Cref{thm:at_most_one_transitive}) to interpret its main component in terms of transitivity of the induced disc game.
Moreover, we show how to handle missing and infinite (when the probability of winning is $0$ or $1$ the logit is infinite) entries (\Cref{sub:optimization_details}).
\citet{Balduzzi2019} proposed to compute the normal decomposition of $2 \mat{P} - 1$ and visualize it as $2$-dim embeddings but did not provide theorems or insights to interpret it.
For ourselves we leverage the symmetric zero-sum game structure, on the other hand, previous works from the pairwise comparison community \citep{Shah2017} relied on threshold singular value decomposition \citep{Chatterjee2015} to provide statistical guarantees on the payoff matrix estimation.

\looseness=-1
While we focus on player strength evaluations for better matchup predictions, a related line of work consists of "only" ranking players without matchup predictions.
\citet{Czarnecki2020} proposed to compute the Nash equilibria of empirical games to cluster players into level sets of ``strength''.
However, when the game has more than two players or is not zero-sum, the computation of the Nash equilibrium is in a class of problem complexity called PPAD-complete~\citep{chen2009settling,daskalakis2009complexity} which is considered intractable.
Motivated by this intractability result~\citet{omidshafiei2019alpha,rashid2021estimating} proposed a tractable ranking technique ($\alpha$-rank) theoretically grounded in the dynamical system theory.


\section{BACKGROUND}
\label{sec:background}
\subsection{Symmetric Zero-Sum Games}
\label{sub:sym_games}
%
In this section, we first define symmetric zero-sum games (\Cref{def:empirical_game}), fully transitive and cyclic games (\Cref{def:cycli_trans}).
Then we provide two examples of games, the Elo \emph{game} (\Cref{ex:elo_game}) and the disc game (\Cref{ex:disc_game}).
%
%
\begin{definition}
\label{def:empirical_game}
    We define a symmetric zero-sum game through
        a probability matrix
        $\payoffprobmat \in \sR^{n \times n}$,
        such that for all $i, j \in [n]$,
        $\payoffprob{i}{j} \in [0,1]$
        and
        $\payoffprob{i}{j}
        = 1
        - \payoffprob{j}{i}$.
        We call a player an index $i$ of this matrix.
\end{definition}
Using \Cref{def:empirical_game}
\begin{enumerate*}[series = tobecont, itemjoin = \quad, label=(\roman*)]
    \item $\payoffprob{i}{j} >0.5$ can be seen as $i \text{ beats } j$
    \item $\payoffprob{i}{j} <0.5$ corresponds to $i \text{ is beaten by } j$
    \item $\payoffprob{i}{j} =0.5$ is a tie.
\end{enumerate*}
For many real-world games, one can have access to large databases of human game outcomes: the symmetric zero-sum game payoff matrix can be estimated empirically.
Along with this paper, we will use empirical payoff and symmetric zero-sum games interchangeably.
We argue that real-world zero-sum games lie on a spectrum between being fully transitive (without any cycle) and being fully cyclic ( all the players belong to the same cycle).
\begin{definition}\emph{(Fully Cyclic and Transitive).}
    \label{def:cycli_trans}
    The game $\payoffprobmat$
    is said to be \emph{fully transitive} if
    for all players $i, j, k$, if
    $\payoffprob{i}{j} > 0.5$
    and
    $\payoffprob{j}{k} > 0.5$
    then
    $\payoffprob{i}{k} > 0.5$.
    The game $\payoffprobmat$ is said to be \emph{fully cyclic} if there exists an ordering $\gamma$ such that
    $
    \payoffprob{\gamma(1)}{\gamma(2)} > 0.5,
    \dots,
    \payoffprob{\gamma(n-1)}{\gamma(n)} > 0.5
    $
    and
    $\payoffprob{\gamma(n)}{\gamma(1)} > 0.5$.
\end{definition}
\begin{figure}[tb]
    \centering
    \includegraphics[width=1\linewidth]{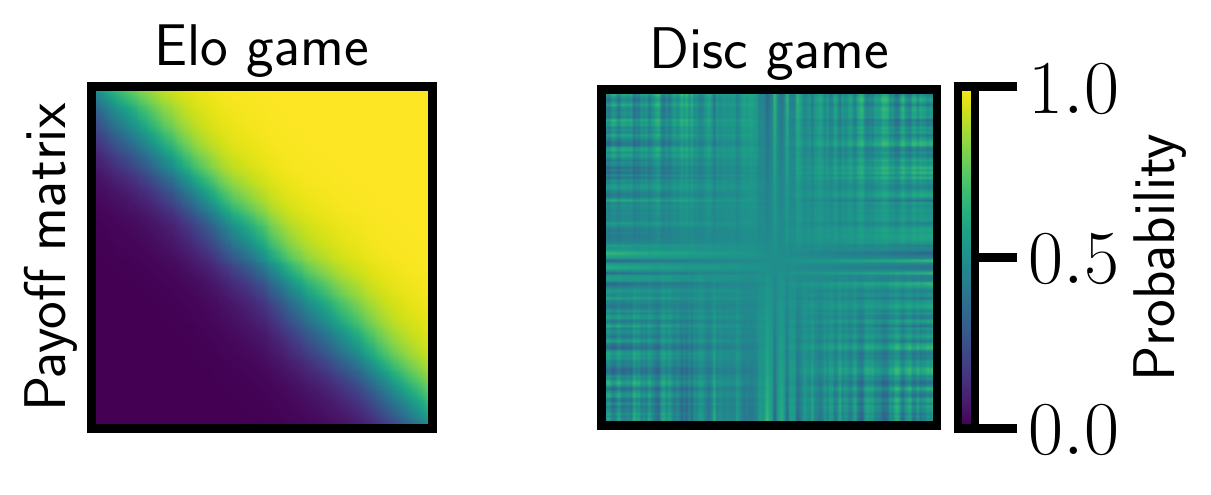}
    \caption{
        \textbf{Elo and Cyclic Disc Games.}
        Payoff matrices for an Elo (left) and a disc (right) game.
        For the Elo game (left) one can see that player $1$ beats everyone: $\payoffprob{1}{j} > 1/2$ for all $j \in [n] \backslash \{ 1 \}$.
        Player $2$ beats everyone except player $1$, etc.
        Conversely, there is no total order for the disc game (right).
        }
        \label{fig:intro_elo_disc}
\end{figure}
The prototypical examples of fully transitive games are Elo games
(\Cref{ex:elo_game}, \Cref{fig:intro_elo_disc} left): each player is assigned one score (the Elo),
and the probability of the outcome between two players is an increasing function (a sigmoid) of the difference in the Elo scores.
\begin{example}\emph{(Elo Game~\citealt{Balduzzi2019}).}
    Let $\vect{u} \in \sR^n\,,$
    $\payoffprob{i}{j} = \sigma(\player{i} - \player{j})$, for all $i, j \in [n]$.
    \label{ex:elo_game}
\end{example}
One can easily show that Elo games (\Cref{ex:elo_game}) are fully transitive.
However, the ''reciprocate'' is false in general: in \Cref{sub:elo} we provide an example (\Cref{ex:trans_non_elo}) of a fully transitive game for which the Elo score fails to correctly rank players.
On the other side of the spectrum, the typical example of a fully cyclic game is the cyclic disc game (\Cref{ex:disc_game}, \Cref{fig:intro_elo_disc} right), where each player $i$ is assigned two scores
$(\player{i}, \playerv{i})$.
\begin{example}\emph{(Disc Game,~\citealt{Balduzzi2019}).}
     Let $\vect{u}, \vect{v} \in \mathbb{R}^n$,
    $\payoffprob{i}{j} = \sigma(\player{i} \playerv{j} - \playerv{i} \player{j})\,,\,$ for all  $i, j  \in [n]$.
    \label{ex:disc_game}
\end{example}

For example, if for all $i \in [n]$, $\player{i}=1$, the disc game is an Elo game (\Cref{ex:elo_game}) and is transitive.
One the other hand, if $\vect{u} = (\cos \tfrac{2\pi i}n)_{i \in [n]}, \vect{v}= (\sin\tfrac{2\pi i}n)_{i \in [n]}$, then the disc game is fully cyclic (\Cref{fig:intro_elo_disc}, right).
The main contribution of \Cref{sub:cstr_disc_game} is to show that a disc game can be nothing else but fully transitive or fully cyclic (\Cref{prop:disc_game}).
%
\subsection{The Elo Score and its Limitations}
\label{sub:elo}
In this section, we recall the definition of the stationary Elo score (\Cref{eq:elo_stationary}).
Then we recall one usual issue with the Elo score.
Finally, we recall that the Elo score can fail on some transitive games (\Cref{ex:trans_non_elo,fig:elo_break}).
\paragraph{Recalls on the Elo score}
For zero-sum symmetric empirical games, \citet{elo1978rating} proposed a rating system able to predict the probability of the outcome of a game between two agents. Given two agents $i$ and $j$,
with a respective Elo score of $\player{i}$ and $\player{j}$,
the probability of $i$ beating $j$ under the Elo model is
    $\sP(i \text{ beats } j)
    =
    \sigma(\alpha (\player{i} - \player{j}))$,
where $\alpha > 0$ is a scaling factor which brings the values of $\vect{u}$ in a range which is easier to grasp for humans
(for simplicity  $\alpha$ is set to $1$).
In a stationary regime, for an empirical payoff matrix $\mat{P}$, the\footnote{The Elo score is not unique and is defined up to a constant.} Elo score $\player{i}^{\mathrm{Elo}}$ of each player $i$ is defined as following \citep[Prop. 1]{Balduzzi2018}, with $\mathrm{bce}$ being the binary cross entropy,
%
\begin{align}
    \label{eq:elo_stationary}
    \vect{u}^{\mathrm{Elo}}
    =
    \argmin_{\vect{u}} \sum_{i, j}
    \mathrm{bce}
    (\payoffprob{i}{j}, \sigma(\player{i} - \player{j}))
    \enspace.
\end{align}
\paragraph{Issues with the Elo score}
\looseness=-1
A first issue with the Elo score is that it assumes that the modeled game is additive (in the $\logit$ space): $\logit(\payoffprob{i}{j}) + \logit(\payoffprob{j}{k}) = \logit(\payoffprob{i}{k})$.
Hence the modeled game should be transitive, which is not always the case for real-world games (such as rock-paper-scissor or StarCraft II, for instance).
Another issue with the Elo score is the following:
even if a game is transitive, the Elo score can "fail" at ranking the players correctly.
Indeed, there are some situations where the Elo score, a single scalar variable, is not expressive enough to predict the outcome of future confrontations.

\begin{example}\label{ex:trans_non_elo}
    Here we define a family of three-player transitive games for all $\gamma, \delta \in (0.5, 1]$.
    Contrary to Elo games (\Cref{ex:elo_game}), outcome probabilities might be non-additive:
    \begin{align}
        \label{eq:trans_non_elo}
        \mat{P}^{(\gamma, \delta)} = \begin{pmatrix}
            0.5 & \gamma & \gamma\\
            1 - \gamma & 0.5 & \delta \\
            1 - \gamma & 1 - \delta & 0.5
        \end{pmatrix}
        \enspace.
        \nonumber
    \end{align}
\end{example}
\Cref{ex:trans_non_elo} describes the payoff matrix of a game that is transitive for $\gamma, \delta \in (0.5, 1]$, however when $\gamma$ is close to $0.5$ and $\delta$ is close to $1$---i.e., the second player slightly loses against the first one and significantly wins against the third one---Elo score fails to assign scores which yield correct matchup predictions between players.
\Cref{fig:elo_break} displays the set of values $\gamma, \delta$ for which the Elo score fails (in red) and succeeds (in green) to correctly estimate the probability of winning between the first and the second players of the game $\mat{P}^{(\gamma, \delta)}$ (\Cref{ex:trans_non_elo}).
Despite the game being transitive (player $1$ beats player $2$ and $3$ and player $2$ beats player $3$), there exists a significant range of values $\gamma, \delta$, for which the Elo score assigns a larger score to player $2$ than player $1$, and thus wrongly predicts the outcome of the confrontation.
%
\begin{figure}[tb]
        {    \centering
    \includegraphics[width=0.95\linewidth]{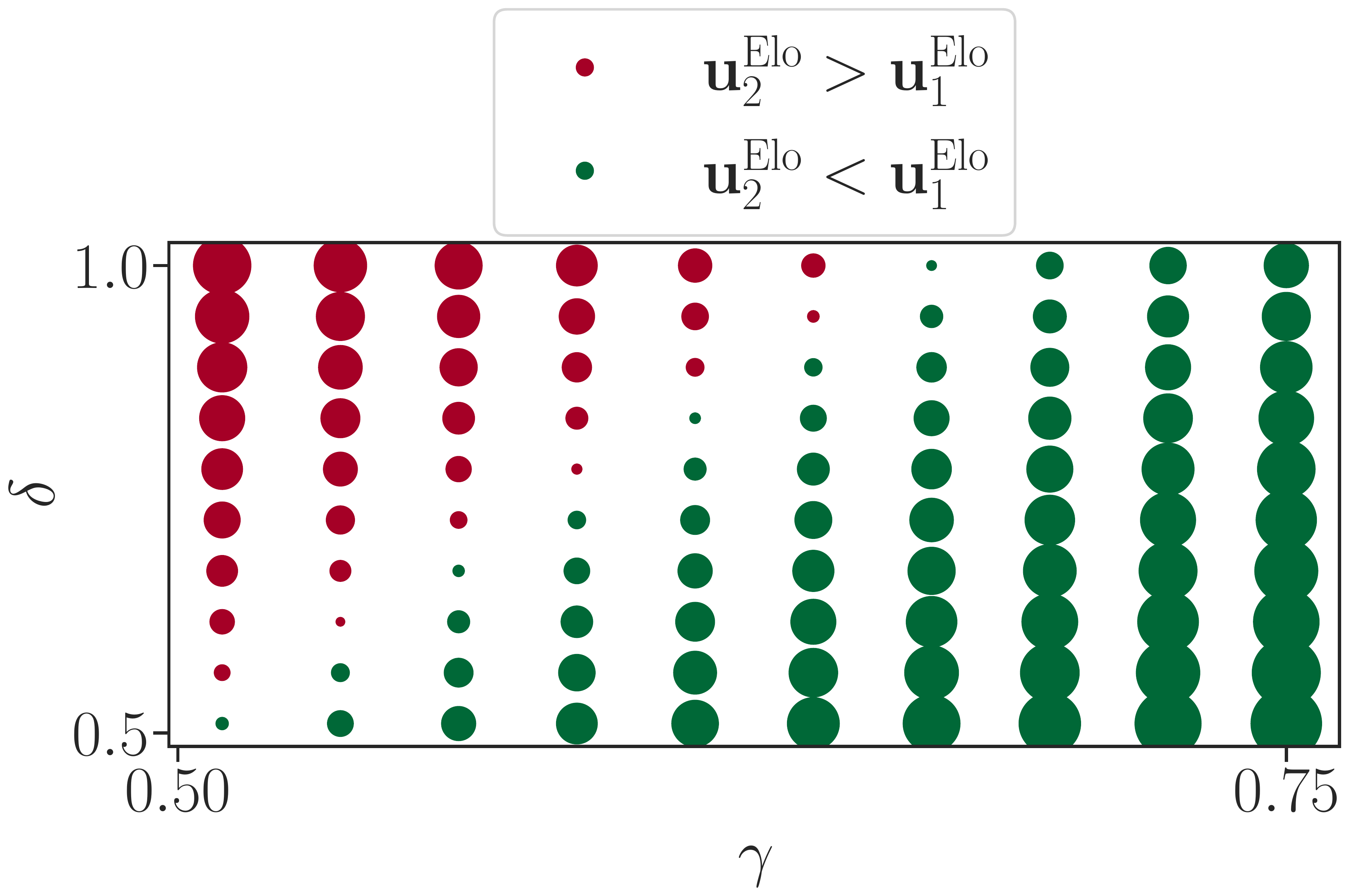}
    }
    \caption{\textbf{Elo Score Fails to Rank Players for Some Transitive Games.}
    This figure displays the difference $\player{1}^{\mathrm{Elo}} - \player{2}^{\mathrm{Elo}}$ between the Elo scores of $\player{1}^{\mathrm{Elo}}$ of player $1$, and $\player{2}^{\mathrm{Elo}}$ of player $2$ (computed using the stationary Elo score in \Cref{eq:elo_stationary}).
    The difference of the Elo scores is displayed for multiple probability matrices $\mat{P}^{(\gamma, \delta)}$ of the transitive game (\Cref{ex:trans_non_elo}).
    Red dots indicate that $\player{2}^{\mathrm{Elo}} > \player{1}^{\mathrm{Elo}}$ and green dots indicate $\player{1}^{\mathrm{Elo}} > \player{2}^{\mathrm{Elo}}$,
    the size of the dots is proportional to $|\player{1}^{\mathrm{Elo}} - \player{2}^{\mathrm{Elo}}|$.
    }
    \label{fig:elo_break}
\end{figure}

%
\paragraph{Connection with Stochastically Transitive Models}
 Bradley-Terry-Luce and Thurstone models \citep{tutz1986bradley,atkinson1998asian,Negahban2012,Shah2017} are a generalization of the Elo model.
There are used to stochastically approximate transitive models of a specific type of what we called symmetric zero-sum games with a strong transitive component.
Similar limitations with the Elo score that the ones we present in \Cref{ex:trans_non_elo} were previously mentioned by \citet[Figure 1a]{Shah2017}.
Moreover, \citet{Shah2017} showed that there exist some transitive matrices (called SST matrices) which are poorly approximated by one-dimensional parametric models.
This result motivates our main contribution, which corresponds to a multi-dimensional parametric model approximating a matrix that corresponds to a symmetric zero-sum game.

\looseness=-1
Failures of the Elo score to correctly `rank' players for a transitive game (\Cref{ex:trans_non_elo,fig:elo_break})
call for models which can handle a larger class of real-world games, going beyond Elo games.
In \Cref{sec:proposed_approach} we propose a \ours, which can correctly rank players on \Cref{ex:trans_non_elo} (see \Cref{fig:extended_elo_works}).
%

\section{PROPOSED APPROACH}
\label{sec:proposed_approach}
%
We now present our main contribution, which is threefold:
\begin{itemize}[noitemsep,itemjoin = \quad,topsep=0pt,parsep=0pt,partopsep=0pt, leftmargin=*]
    \item
    First in \Cref{sub:cstr_disc_game} we thoroughly study the disc game, $\payoffprob{i}{j}
    =
    \sigma(\player{i} \playerv{j} - \playerv{i} \player{j}) \,, \: i,j \in [n]$  (\Cref{ex:disc_game}).
    We show that depending on the values of $\vect{u}$ and $\vect{v}$,
    the disc game is either fully transitive or fully cyclic (\Cref{prop:disc_game}).
    \item Then, in \Cref{sec:theoretical_analysis} we show that any empirical matrix from real-world zeros-sum games can be decomposed as a sum of disc games, with \emph{at most one transitive disc game} (\Cref{thm:at_most_one_transitive}).
    This motivates our disc rating system which corresponds to the extraction of the empirical matrix's main component (disc game).
    If this main component corresponds to a transitive (resp. cyclic) disc game, then the original game is mainly transitive (resp. cyclic).
    \item Finally in \Cref{sub:optimization_details} we provide the optimization details needed to compute the proposed scores (\Cref{alg:otho_decomp}).
\end{itemize}
Given a payoff matrix $\mat{P}$ (\Cref{def:empirical_game}), the following diagram summarizes the proposed approach
\begin{align}
    \mat{P}
    \xrightarrow[]{\mathrm{logit}} \underbrace{\mat{A}}_{\text{skew-symmetric}}
    \xrightarrow[\text{decomposition (Thm. \ref{thm:antisym_pca})}]{\text{Truncated normal}} \mat{\hat A}
    \xrightarrow[]{\mathrm{sigmoid}} \mat{\hat P}.
    \nonumber
\end{align}
\subsection{Detailed Study of the Disc Game}
\label{sub:cstr_disc_game}
\Cref{sub:sym_games} provides examples of $\vect{u}, \vect{v} \in \mathbb{R}^n$ values for which the disc game (\Cref{ex:disc_game}) is \emph{fully transitive} or \emph{fully cyclic}.
In this section, we show that there are no other options: a disc game is either be \emph{fully cyclic} or \emph{fully transitive}, depending on the of the values of $\vect{u}, \vect{v} \in \mathbb{R}^n$.
\begin{propositionrep}\emph{(A Disc Game is Fully Cyclic or Fully Transitive).}
    \label{prop:disc_game}
    Let $\vect{u}, \vect{v} \in \sR^n$ and
    $\payoffprob{i}{j}
    =
    \sigma(\player{i} \playerv{j} - \playerv{i} \player{j})\,,\; i,j \in [n]$ be a disc game.
    Let $U:= \hull\{(\player{i},\playerv{i})\}$ be the convex hull of the players. If $(0,0)$ is \emph{not} in the border of $U$,\footnote{if it is the case we can, for instance, apply an infinitesimal perturbation on the points $(\player{i},\playerv{i})$.}
    then the disc game is either fully transitive or fully cyclic. Precisely, the disc game is
    \begin{enumerate}
        \item fully cyclic if and only if $(0,0) \in \Int(U)$,
        \item fully transitive if and only if $(0,0) \notin \Int(U)$.
    \end{enumerate}
\end{propositionrep}
\begin{appendixproof}
Let us start this proof by showing that there always exists a cycle a size $3$ in a larger cycle.
\begin{lemma}[A cycle of size $n$ implies a cycle of size 3]
If there exists $\player{1},\ldots,\player{n} \in U$ such that $\player{1} \xrightarrow{\!\mathrm{beats}\!} \ldots  \xrightarrow{\!\mathrm{beats}\!} \player{n} \xrightarrow{\!\mathrm{beats}\!} \player{1}$, then there exists $u,v,w \in U$ such that $u \xrightarrow{\!\mathrm{beats}\!} v \xrightarrow{\!\mathrm{beats}\!} w \xrightarrow{\!\mathrm{beats}\!} u$.
\end{lemma}
\begin{proof}
We will use a recursive argument. It is true for $n=3$. Let us assume it is true for a given $n \in \sN$, then given $\player{1} \xrightarrow{\!\mathrm{beats}\!} \ldots  \xrightarrow{\!\mathrm{beats}\!} \player{n+1} \xrightarrow{\!\mathrm{beats}\!} u_1$ we have two possible cases
\begin{enumerate}
    \item $\player{1} \xrightarrow{\!\mathrm{beats}\!} \player{n}  \xrightarrow{\!\mathrm{beats}\!} \player{n+1} \xrightarrow{\!\mathrm{beats}\!} \player{1}$:
    In that case, we have proved the result.
    \item $\player{1} \xleftarrow{\!\mathrm{beats}\!} \player{n} \xrightarrow{\!\mathrm{beats}\!} \player{n+1} \xrightarrow{\!\mathrm{beats}\!} \player{1}$.
    In that case, it means that  $\player{1} \xrightarrow{\!\mathrm{beats}\!} \ldots  \xrightarrow{\!\mathrm{beats}\!} \vect{u}_{n} \xrightarrow{\!\mathrm{beats}\!} \player{1}$ and thus we can use the recurrence hypothesis that there exists a cycle of size $3$ inside a cycle of size $n$.
\end{enumerate}
\end{proof}

We can now show Proposition~\ref{prop:disc_game}.

The idea of the proof is to use the polar coordinates for all the points $(\player{i},\playerv{i})=(r_i,\theta_i)\in \sR^2$. We can notice that $\payoffprob{i}{j} = \sigma(\player{i} \playerv{j} - \playerv{i} \player{j}) = \sigma(r_i r_j \sin(\theta_i - \theta_j))$. Thus $\payoffprob{i}{j}>1/2$ if and only if $\theta_j + \pi >\theta_i> \theta_j$.
If the game is cyclic it means that there exists a cycle of size 3. Without any loss of generality let us assume $\theta_1=0$. Then we have that $\pi>\theta_2>0$ and $-\pi<\theta_3<0$. Finally the fact that $\theta_3<\theta_2+\pi$ implies that either $\theta_2>\pi/2$ or $\theta_3 <\pi/2$.
Thus this forms a triangle that contains $(0,0)$.
Thus we can show that $(0,0)$ is in the interior of the convex set defined by three points of the cycle. If $(0,0)$ is in the interior of $U$ then we can find a cycle.
\end{appendixproof}
The main takeaway from this result is that the convex hull of the players $U$ is the key object to determine if the disc game is transitive or cyclic.
The dichotomy happens around the origin.
The next proposition states that if a disc game is transitive, then $\vect{u}$ and $\vect{v}$ can be reparametrized.
\begin{propositionrep}\emph{(Reparametrization).}
\label{prop:positive_trans}
Let $\vect{u},\vect{v} \in \sR^n $ and
    $\payoffprob{i}{j}
    :=
    \sigma(\player{i} \playerv{j} - \playerv{i} \player{j})\,,\; i,j\in [n]$ be a \textbf{transitive} disc game.
    Then there exists a reparametrization of the disc game $(\playertilde{i},\playertildev{i})$ such that, for all $i \in [n]$, $ \playertildev{i} > 0$, and
    \begin{equation}
        \sigma(
            \player{i} \playerv{j} - \playerv{i} \player{j}
            )
        = \payoffprob{i}{j}
        =
        \sigma(
            \playertilde{i} \playertildev{j} - \playertildev{i} \playertilde{j})
    \enspace.
    \end{equation}
\end{propositionrep}
\begin{appendixproof}
To prove this result we will use \Cref{prop:disc_game}. Since the disc game is considered transitive and that $(0,0)$ is not in the border of $U$, we know that $(0,0) \notin U = \hull \{(\player{i},\playerv{i})\}$. Thus, by the hyperplane separation theorem~\citep[Example 2.20]{boyd2004convex} there exists a direction $a \in \sR^2\,,\; \|a\|_2=1$ such that, $\langle a,[\player{i},\playerv{i}]\rangle >0\,,\; \forall i \in [n]$. Setting $b := [a_2, -a_1]$, we get that $(a,b)$ is an orthonormal basis of $\sR^2$. Let $(\playertilde{i},{\vect{v}}_i):= (\langle a,[\player{i},\playerv{i}]\rangle, \langle b, [\player{i},\playerv{i}])$ be the coordinate of $[\player{i},\playerv{i}]$ in this new basis. Finally, we just need to remark that $\player{i} \playerv{j} - \playerv{i} \player{j}$ does not depend on the choice of basis. It is because,  $\player{i} \playerv{j} - \playerv{i} \player{j} = \sigma(\|OM_i\|\|OM_j\| \sin \widehat{M_iOM_j})$ where $M_i \in \sR^2$ corresponds to the point with coordinate $[(\player{i},\playerv{i})]$.
\end{appendixproof}
The proof of this result relies on \Cref{prop:disc_game} and the hyperplane separation theorem~\citep[Example 2.20]{boyd2004convex}.
Because the reparametrization $(\tilde {\vect{u}}_i,\playertildev{i})$ yields the same payoffs, it corresponds to the same disc game.
\Cref{prop:positive_trans} means that, without any loss of generality, one can consider that $\playerv{i} >0\,,\; i \in [n]$ for any transitive disc games.

In the next section, we will see that any empirical payoff $\rmP$ (\Cref{def:empirical_game}) can be transformed into a skew-symmetric matrix $\rmA$ and decomposed as a sum of disc games. We will then use~\Cref{prop:disc_game} on the main disc game to assess the transitivity/cyclicity of the original empirical game $\rmA$.
%
\subsection{Disc Decomposition}
\label{sec:theoretical_analysis}
First, we recall a standard result on the decomposition of skew-symmetric matrices (\Cref{thm:antisym_pca}): any real skew-symmetric matrix can be decomposed as a sum of matrices of the form
$\vect{u} \vect{v}^\top - \vect{v} \vect{u}^\top$,
$\vect{u}, \vect{v} \in \sR^n$.
Combining this result with \Cref{prop:disc_game}, we finally show that zero-sum game payoff matrices have at most a single transitive disc game component (\Cref{thm:at_most_one_transitive}).
\begin{theoremrep}\emph{(Normal Decomposition, \citealt{greub1975linear}).}
    \label{thm:antisym_pca}
    Suppose $\mat{A} \in \sR^{n \times n}$ is such that $\payoffanti{}{} = - \payoffanti{}{}^\top$.
    Then, with $k = \floor{n/2}$, there exists
    $\lambda_1 \geq \ldots \geq \lambda_k$
    and
    $(\vect{u}^{(l)}, \vect{v}^{(l)}) \in \sR^n \times \sR^n
    \,,\,1
    \leq
    l
    \leq
    k$,
    such that $(\vect{u}^{(l)}, \vect{v}^{(l)})$ is an orthogonal family and
        $\mat{A}
        =
        \sum_{l=1}^{k}
             (
            \vect{u}^{(l)} \vect{v}^{(l) \top}
            - \vect{v}^{(l)} \vect{u}^{(l) \top} ) $.

\end{theoremrep}
\begin{appendixproof}
    Proof of \Cref{thm:antisym_pca} can be found in \citet[Sec. 2.5]{cassiniAlgebre3} or \citet[\S8.16]{greub1975linear}.
\end{appendixproof}
%
A proof of \Cref{thm:antisym_pca} can be found in \citet[Sec. 2.5]{cassiniAlgebre3}.
\Cref{thm:antisym_pca} is sometimes referred to as the Schur decomposition of skew-symmetric matrices \citep[Prop. 2]{Balduzzi2018}.

In the context of symmetric zero-sum games, the first takeaway of Theorem~\ref{thm:antisym_pca} is that
such a game is a sum of disc games.
The second consequence (\Cref{thm:at_most_one_transitive}) is less obvious but maybe even more important:
among all these disc games \emph{at most one} is transitive.
    \begin{theoremrep}\emph{(An Empirical Game has at most a Single Transitive Disc Game)}.
        \label{thm:at_most_one_transitive}
        Let $\mat{P}$ be the payoff matrix of a symmetric zero-sum game, and let $(\vect{u}^{(l)}, \vect{v}^{(l)}) \in \sR^{n} \times \sR^{n} $ be the normal decomposition of the skew-symmetric matrix $\logit(\mat{P})$ (\Cref{thm:antisym_pca}), then there exists \emph{at most} one pair $(\vect{u}^{(l)}, \vect{v}^{(l)}) \in \sR^{n} \times \sR^{n} $ such that the disc game defined by
                $\payoffprob{i}{j} =
                    \sigma(\player{i}^{(l)} \playerv{j}^{(l)}
                    - \playerv{i}^{(l)} \player{j}^{(l)} )$,
        $i, j \in [n]$,
        is transitive.
    \end{theoremrep}
    \begin{appendixproof}
    \rebutal{
    The proof of \Cref{thm:at_most_one_transitive} relies on
    \begin{itemize}
        \item  The  fact that $(\vect{u}^{(l)}, \vect{v}^{(l)})$ are orthogonal.
        \item The reparametrization property (\Cref{prop:positive_trans}).
    \end{itemize}}
    \rebutal{
    The normal decomposition $(\vect{u}^{(l)}, \vect{v}^{(l)})$ is composed of orthogonal vectors ($\vect{v}^{(l)}$ are orthogonal to each others and to $\vect{u}^{(l)}$).}

    \rebutal{
    Suppose that there exist two transitive pairs of components, $(\vect{u}^{(1)}, \vect{v}^{(1)})$ and $(\vect{u}^{(2)}, \vect{v}^{(2)})$. Using  \Cref{prop:positive_trans} we consider the reparametrization such that $\playerv{i}^{(1)} > 0$, for all $i \in [n]$, and $\playerv{i}^{(2)} > 0$, for all $i \in [n]$, hence one has $\vect{v}^{(1) \top } \vect{v}^{(2)} > 0$.
    In addition, since $\vect{v}^{(1)}$ and $\vect{v}^{(2)}$ are orthogonal we have that $\playerv{i}^{(1) \top } \playerv{i}^{(2)} = 0$, which yields a contradiction.
    Hence there exists at most one transitive pair of components.}
    \end{appendixproof}
\Cref{thm:at_most_one_transitive} provides us with multiple insights.
If one has access to an empirical game $\mat{P}$, and one can compute the normal decomposition (\Cref{thm:at_most_one_transitive}) of the skew-symmetric matrix $\logit(\mat{P})$,
then the largest value $\lambda_1$ and its associated vectors $(\vect{u}^{(1)}, \vect{v}^{(1)})$ encapsulate information on the largest component.
If the disc game associated with the largest component is transitive, we will say that the considered empirical game based on $\mat{P}$ is transitive.
Note that \Cref{thm:at_most_one_transitive} can be generalized to any invertible functions $f$ which transforms zero-sum game payoff matrices into skew-symmetric matrices, such as for instance $f:\mat{P} \mapsto 2 \mat{P} - 1$.

\Cref{sub:optimization_details} details main component computation when dealing with real data (\emph{e.g.}, missing and inexact entries).
%
\subsection{Computational Details}
\label{sub:optimization_details}
%
%
\begin{algorithm}[t]
    \caption{Alternate Minimization}
    \label{alg:altmin}
    \SetKwInOut{Input}{input}
    \SetKwInOut{Init}{init}
    \SetKwInOut{Parameter}{param}
    \Input{$\mat{A} \in \sR^{n \times  n},
        \mat{us}, \mat{vs} \in \sR^{n \times  k}, l \in \mathbb{N}$}
    \Init{$\vect{u}, \vect{v} \neq 0_{n}$}
    \For{$k = 1, 2 \dots, $}
    {
    $\vect{u} \leftarrow \argmin_{\vect{u}}
    \mathcal{L}(\mat{P}, \mat{A} + \vect{u} \vect{v}^\top - \vect{v} \vect{u}^\top)
    $
    \\
    $
    +
    \sum_{m=1}^{l-1} \frac{\langle \vect{u}, \vect{vs}_{:m} \rangle^2}{\normin{\vect{vs}_{:m} }^2 }
    +
    \sum_{m=1}^{l-1} \frac{\langle \vect{u}, \vect{us}_{:m} \rangle^2}{\normin{\vect{us}_{:m} }^2 }
    $

    $\vect{v} \leftarrow
    \argmin_{\vect{v}}
     \mathcal{L}(\mat{P}, \mat{A} + \vect{u} \vect{v}^\top - \vect{v} \vect{u}^\top)
    + \tfrac{\langle \vect{u}, \vect{v} \rangle^2}{\normin{\vect{u}}^2 }
    $
    \\
    $
    +
    \sum_{m=1}^{l-1} \tfrac{\langle \vect{v}, \vect{vs}_{:m} \rangle^2}{\normin{\vect{vs}_{:m} }^2 }
    +
    \sum_{m=1}^{l-1} \tfrac{\langle \vect{v}, \vect{us}_{:m} \rangle^2}{\normin{\vect{us}_{:m} }^2 }
    $
    }
\Return{$\vect{u}, \vect{v}$}
\end{algorithm}
%
%
Now we provide the details to compute the proposed score.
Consider an empirical game $(\payoffprob{i}{j})_{\leq i,j \leq n}$.
One can compute the decomposition of $\logit(\vect{P})$ from \Cref{thm:antisym_pca} by solving the following optimization problem, with $\mathcal{L} : \sR \times \sR \rightarrow \sR$, $\mathcal{L}: x, \hat x \mapsto \norm{\logit(x) - \hat x}^2$
\begin{align}
    &\min_{\vect{u}^{(k)}, \vect{v}^{(l)} \in \sR^n}
    f  (\vect{u}, \vect{v})
    \triangleq
     \mathcal{L} (\mat{P}, {\textstyle\sum_{l = 1}^k}\vect{u}^{(l)} {\vect{v}^{(l)}}^\top - \vect{v}^{(l)} {\vect{u}^{(l)}}^\top)
     \nonumber
    \\
     & \quad \quad \quad \quad \quad \quad \triangleq
    \sum_{i, j,l}
     \mathcal{L}(\payoffprob{i}{j}, \player{i}^{(l)} \playerv{j}^{(l)} - \playerv{i}^{(l)} \player{j}^{(l)})
    \enspace,
    \label{eq:opt_pb_decomposition_gen}
    \\
    &
    \text{s. t.} \quad \langle \vect{u}^{(l)}, \vect{u}^{(m)}\rangle  = \langle \vect{v}^{(l)}, \vect{v}^{(m)}\rangle  = 0 \enspace,
    \\
    &
    \quad \quad \: \: \:
    \langle \vect{u}^{(l)}, \vect{v}^{(l)}\rangle  = 0 \,,\quad
    \quad 1\leq l < m \leq k\, \enspace .
\end{align}
The main challenge of this optimization problem is to maintain the orthogonality constraint between the components.
We propose~\Cref{alg:otho_decomp} which sequentially finds the $l^{th}$ main components $(\vect{u}^{(l)}, \vect{v}^{(l)})$ for $l \in [k]$.
The $l^{th}$ pair of components is found with~\Cref{alg:altmin} which maintains the orthogonality constraints by using a penalty term (inspired by the algorithm of~\citet{Gemp2021} which computes PCA for large-scale problems).
Note that \emph{in practice, empirical probability matrices can have $0$ or $1$ entries (for instance if a player $i$ always loses again player $j$), for which $\logit$ is not defined.}
That's why in the experiments (\Cref{sec:experiments}) we use the following loss function
 $\mathcal{L} : x, \hat x \mapsto \mathrm{bce}(x, \sigma(\hat x))$ (where $\mathrm{bce}$ is the binary cross entropy).
 \paragraph{Missing Entries}
 Payoff matrices coming from real-world games played by humans usually contain missing entries: one does not have access to the matchups between all the players.
 For instance, on the \texttt{Lichess} website, there usually are few confrontations between low-ranked and high-ranked players.
 In other words, one only partially has access to $\payoffprob{i}{j}$, for $(i, j)$ in a given set of pairs of players $\mathcal{D}^{\mathrm{obs}} \subset  [n] \times [n]$.
 Note that the optimization problem formulations based on \Cref{eq:opt_pb_decomposition_gen} can handle missing entries (by summing only on the available entries) as in \citet{CandesRecht2008,CandesPlan2009}.
 With $k=1$ pair of components, instead of the problem defined in \Cref{eq_app:opt_pb_decomposition_gen}, the proposed \ours with a partial set of observation $\mathcal{D}^{\mathrm{obs}}$ reads:
 \begin{align}
     \argmin_{
         \vect{u}, \vect{v},
         \text{ s.t. } \vect{u}^\top \vect{v} = 0}
     \sum_{(i, j) \in \mathcal{D}^{\mathrm{obs}}}
      \mathcal{L}(\payoffprob{i}{j}, \player{i} \playerv{j} - \playerv{i} \player{j})
     \enspace.
     \nonumber
 \end{align}

\sloppy
\paragraph{Interpretation of the Disc Decomposition in the Case $k=1$}
Let $(\vect{u}^{\mathrm{Disc}}, \vect{v}^{\mathrm{Disc}})$ be the first pair ou component of the disc decomposition (\Cref{eq:opt_pb_decomposition_gen}).
The probability of player $i$ beating player $j$ is given by
\begin{align}
    \mathbb{P}(i \text{ beats } j) = \sigma(\player{i}^{\mathrm{Disc}} \playerv{j}^{\mathrm{Disc}} - \playerv{i}^{\mathrm{Disc}} \player{j}^{\mathrm{Disc}})
    \enspace.
\end{align}
If it occurs that for all $j\in [n]$, $\playerv{j}^{\mathrm{Disc}} > 0$, which corresponds to a transitive game (see \Cref{thm:at_most_one_transitive} and \Cref{prop:positive_trans}), then the outcome probability can be written as
$\sigma(\player{i}^{\mathrm{Disc}} \playerv{j}^{\mathrm{Disc}} - \playerv{i}^{\mathrm{Disc}} \player{j}^{\mathrm{Disc}})=
\sigma  (\playerv{i}^{\mathrm{Disc}} \playerv{j}^{\mathrm{Disc}}  ( \player{i}^{\mathrm{Disc}} / \playerv{i}^{\mathrm{Disc}} -  \player{j}^{\mathrm{Disc}} / \playerv{j}^{\mathrm{Disc}} ) ) $.
For a player $i$, the ratio
$ \playertilde{i}^{\mathrm{Disc}} \triangleq \player{i}^{\mathrm{Disc}} / \playerv{i}^{\mathrm{Disc}}$
 can be interpreted as its \emph{strength}, and $v_i$ as its \emph{consistency} to beat lower-rated players (and be beaten by higher-rated players):
 \begin{align}
    \mathbb{P}(i \text{ beats } j)
    =
    \sigma(\playerv{i}^{\mathrm{Disc}} \playerv{j}^{\mathrm{Disc}}
    (\playertilde{i}^{\mathrm{Disc}} - \playertilde{j}^{\mathrm{Disc}}))
    \enspace.
\end{align}
If $\playertilde{i}^{\mathrm{Disc}} > \playertilde{j}^{\mathrm{Disc}}$ then $i$ beats $j$.
 The larger $\playerv{i}^{\mathrm{Disc}}$, the larger the probability to win against a lower-rated player ($\playertilde{i}^{\mathrm{Disc}}  > \playertilde{j}^{\mathrm{Disc}} $), but the larger the probability to lose against a higher rated player ($\playertilde{i}^{\mathrm{Disc}}  < \playertilde{j}^{\mathrm{Disc}} $).
 This consistency score can be seen as a way to correct the Elo score which implicitly assumes that the strength of each player is linearly comparable in the logit space.
 For instance, one implicit bias of the Elo score is that if $i$ beats $j$ with probability $\payoffprob{i}{j}$ and $j$ beats $k$ with probability $\payoffprob{j}{k}$ then $i$ beats $k$ with probability
 $\logit(\payoffprob{i}{k}) =
 \logit(\payoffprob{i}{j})
 +
 \logit(\payoffprob{j}{k})$.
 As illustrated in \Cref{ex:trans_non_elo} if such a property of additivity in the logit space is not occurring in the data, the Elo score may have trouble even predicting the right ranking for the players (\Cref{fig:elo_break}). Conversely, when including a notion of \emph{consistency} in our score we manage to correctly predict the ranking between the players in~\Cref{ex:trans_non_elo} (see~\Cref{fig:extended_elo_works}.)

%
\begin{algorithm}[tb]
    \caption{Compute Disc Decomposition}
    \label{alg:otho_decomp}
    \SetKwInOut{Input}{input}
        \SetKwInOut{Init}{init}
        \SetKwInOut{Parameter}{param}
         \Input{$k \in \mathbb{N}$ \textcolor{blue}{($\#$ of pairs of components)}}
         \Init{$\mat{A} = 0_{n \times n}$, \\
         $\mat{us} = 0_{n \times  k}$,
         $\mat{vs} = 0_{n \times  k}$
         }
        \For{$l = 1, 2, \dots, k$ }{
            $\vect{u}^{(l)}, \vect{v}^{(l)} = \mathrm{\Cref{alg:altmin}}(\mat{A}, \mat{us}, \mat{vs}, l)$

            $\mat{A} \leftarrow \mat{A} + \vect{u}^{(l)}  \vect{v}^{(l) \top}  - \vect{v}^{(l)} \vect{u}^{(l) \top} $

            $\mat{us}_{: l}, \mat{vs}_{: l}
            = \vect{u}^{(l)}, \vect{v}^{(l)}$
        }
    \Return{$(\vect{u}^{(l)}, \vect{v}^{(l)})_{1 \leq l \leq k}$}
\end{algorithm}
\begin{figure}[tb]
\centering
\includegraphics[width=0.95\linewidth]{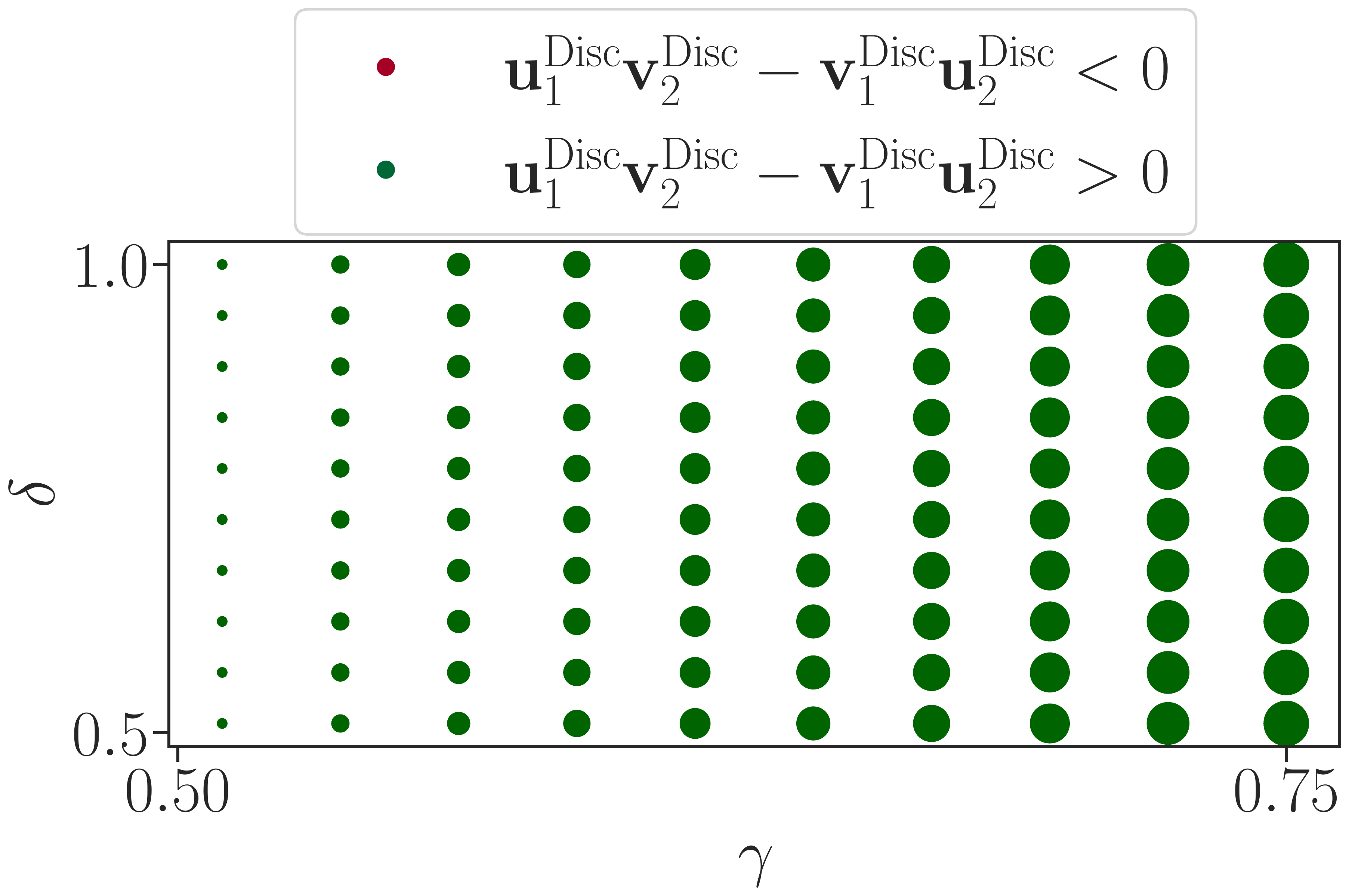}
\caption{
    \textbf{Disc Decomposition Manages to Rank Players for some Transitive Games.}
    Let
    $\vect{u}^{\mathrm{Disc}},
    \vect{v}^{\mathrm{Disc}}$ be the first pair of components computed
    using the stationary \ours ($k=1$ in \Cref{eq:opt_pb_decomposition_gen}).
    The logit of the probability of winning,
    $\player{1}^{\mathrm{Disc}} \playerv{2}^{\mathrm{Disc}} - \playerv{1}^{\mathrm{Disc}} \player{2}^{\mathrm{Disc}}$,  is displayed for multiples probability matrices of the transitive game (\Cref{ex:trans_non_elo}).
    Red dots indicate that $ \player{1}^{\mathrm{Disc}} \playerv{2}^{\mathrm{Disc}} - \playerv{1}^{\mathrm{Disc}} \player{2}^{\mathrm{Disc}} <0$, and green dots $ \player{1}^{\mathrm{Disc}} \playerv{2}^{\mathrm{Disc}} - \playerv{1}^{\mathrm{Disc}} \player{2}^{\mathrm{Disc}} > 0$.
    The size of the dots is proportional to $ |\player{1}^{\mathrm{Disc}} \playerv{2}^{\mathrm{Disc}} - \playerv{1}^{\mathrm{Disc}} \player{2}^{\mathrm{Disc}}|$.
    Contrary to the Elo score (\Cref{fig:elo_break}), the proposed \ours can correctly rank players $1$ and $2$.
}
\label{fig:extended_elo_works}
\end{figure}
\begin{figure*}[t]
\centering
\includegraphics[width=0.9\linewidth]{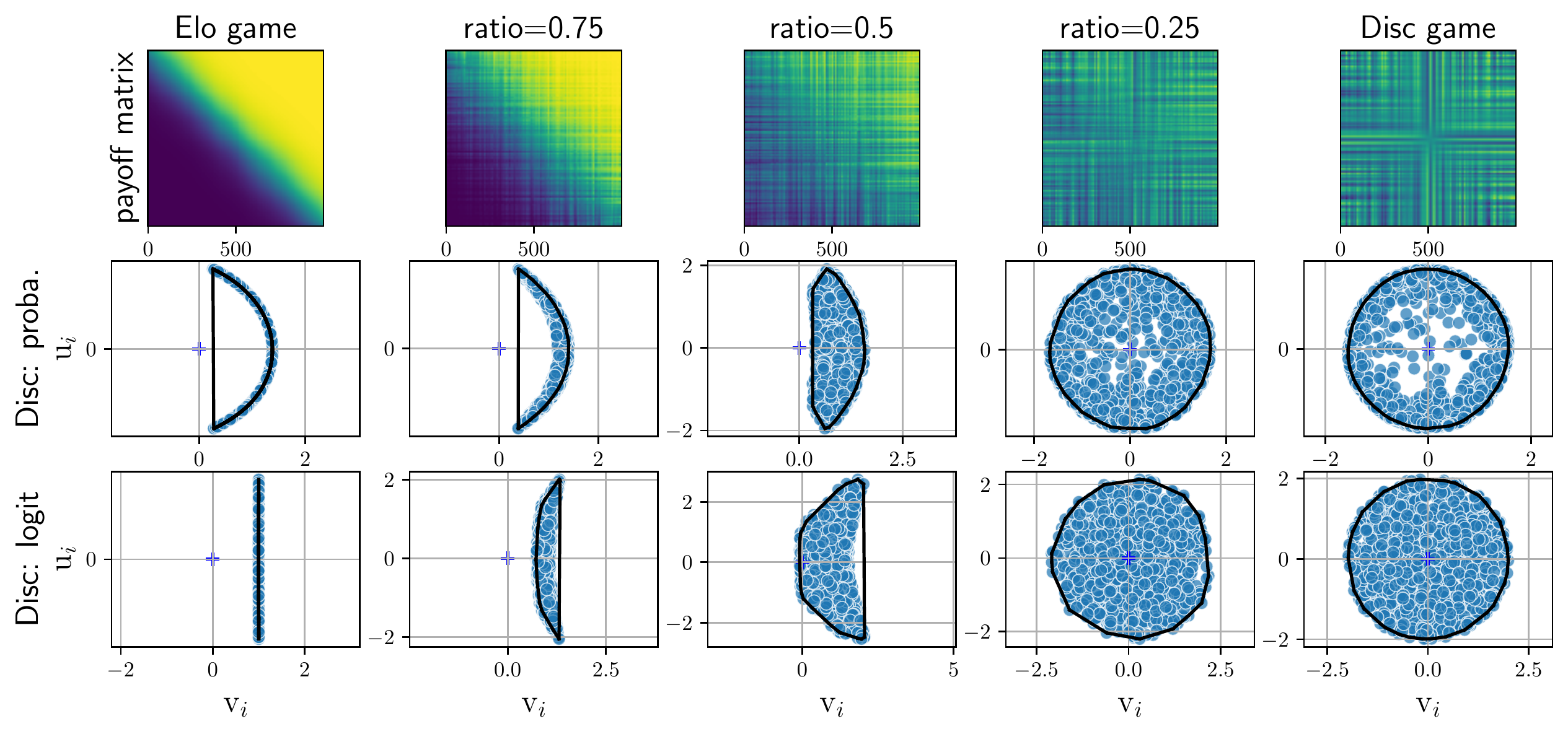}
\caption{\textbf{From Elo to Disc Game.}
Payoff matrices (top) and visualization using Algorithm~\ref{alg:altmin} in the probability space (middle) and logit space (bottom), for multiple games with payoff
$\logit(\mat{P})
=
\mathrm{ratio} \cdot \logit(\mat{P}^{\mathrm{Elo}})
+ (1 - \mathrm{ratio}) \cdot \logit(\mat{P}^{\mathrm{Disc}})
$.}
\label{fig:elo_to_disc}
\end{figure*}
\paragraph{Online Update}
The main focus of this work is to study the intrinsic characteristics of the game at a given time $t$.
This is opposed to the line of work that considers the games' sequential aspect, quantifying the Elo score uncertainty with a Bayesian model (TrueSkill, \citealt{herbrich2006trueskill}), or relying on the more recent games \citep{Sismanis2010}.
However, an update rule can still be derived and interpreted for the proposed \ours.
The stationary version of the proposed \ours is written in \Cref{eq:opt_pb_decomposition_gen}, the online version \citep{Jabin2015} of the \ours writes, with a step size
$\eta >0$,  $\mathrm{S}_{i, j}^t$
the result of the confrontation between players $i$ and $j$,
and
$\payoffprobhat{i}{j}^t
=
\sigma(\playerv{i}^{t} \playerv{j}^{t}(\playertilde{i}^t - \playertilde{j}^t ) )$:
\begin{align}
    \playertilde{i}^{t+1}
    &=
    \playertilde{i}^{t}
    + \eta (\mathrm{S}_{i, j}^t - \payoffprobhat{i}{j}^t ) \playerv{i}^{t} \playerv{j}^{t}
    \enspace,
    \\
    \playerv{i}^{t+1}
    &=
    \playerv{i}^{t}
    + \eta (\mathrm{S}_{i, j}^t - \payoffprobhat{i}{j}^t) \playerv{j}^{t}
    (\playertilde{i}^{t} - \playertilde{j}^{t})
    \enspace.
\end{align}
As for the usual Elo score, the increase in the "strength"
$\playertilde{i}$
after a confrontation is proportional to
$\mathrm{S}_{i, j}^t - \payoffprobhat{i}{j}^t$.
The "consistency" $\playerv{i}$ increases when beating lower-rated players and decreases when losing against lower-rated players. Informally, this quantity encompasses how much one should trust the current Elo score and scales up or down the ``strength'' update.
%

\section{EXPERIMENTS}
\label{sec:experiments}

\begin{table*}[tb]\centering
  \caption{Prediction performance (MSE)
  on unseen data interpolating from a pure Elo game to a disc game: $\logit(\mat{P})
  =
  \mathrm{ratio} \cdot \logit(\mat{P}^{\mathrm{Elo}})
  + (1 - \mathrm{ratio}) \cdot \logit(\mat{P}^{\mathrm{Disc}})
  $.}
  \label{tab:prediction_synthetic}
  \resizebox{1 \textwidth}{!}{
  \large
  \begin{tabular}{cc *{5}{|cc}}
      \toprule
      &
      &  \multicolumn{2}{c}{Elo game}
      & \multicolumn{2}{c}{$\mathrm{ratio}=0.75$}
      & \multicolumn{2}{c}{$\mathrm{ratio}=0.5$ }
      & \multicolumn{2}{c}{$\mathrm{ratio}=0.25$ }
      & \multicolumn{2}{c}{Disc game}
      \\
     Model & \# Param.
     &  Train & Test
     &  Train & Test
     &  Train & Test
     &  Train & Test
     &  Train & Test
     \\
     Elo  & $n$ &
     \num{1.6e-10} & \num{1.6e-10} &
     \num{1.0e-02} & \num{1.0e-02} &
     \num{3.6e-02} & \num{3.6e-02} &
     \num{6.3e-02} & \num{6.3e-02} &
     \num{7.9e-02} & \num{7.9e-02}
     \\
     Elo ++  & $n$ &
     \num{2.1e-11} & \num[math-rm=\mathbf]{2.1e-11} &
     \num{1.0e-02} & \num{1.0e-02} &
     \num{3.6e-02} & \num{3.6e-02} &
     \num{6.3e-02} & \num{6.3e-02} &
     \num{7.9e-02} & \num{7.9e-02}
     \\
     \rowcolor{green!20}
     Disc decomposition (ours, k=1)  & $2n$ &
     \num{2.1e-10} & \num{2.2e-10} &
     \num{9.7e-03} & \num{9.8e-03} &
     \num{3.4e-02} & \num{3.4e-02} &
     \num{1.2e-02} & \num{1.2e-02} &
     \num{2.5e-06} & \num{2.6e-06}
     \\
     \citet{Balduzzi2019}  & $2n$ &
     \num{2.7e-03} & \num{2.7e-03} &
     \num{1.0e-02} & \num{1.0e-02} &
     \num{3.4e-02} & \num{3.5e-02} &
     \num{1.4e-02} & \num{1.4e-02} &
     \num{4.5e-03} & \num{4.6e-03}
     \\
     $m$-Elo (\citealt{Balduzzi2018}, k=1)  & $3n$ &
     \num{2.7e-03} & \num{2.7e-03} &
     \num{1.0e-02} & \num{1.0e-02} &
     \num{3.4e-02} & \num{3.5e-02} &
     \num{1.4e-02} & \num{1.4e-02} &
     \num{4.5e-03} & \num{4.6e-03}
     \\
     \citet{Balduzzi2019}  & $4n$ &
     \num{1.3e-04} & \num{1.3e-04} &
     \num{3.9e-03} & \num{4.0e-03} &
     \num{5.8e-03} & \num{6.0e-03} &
     \num{4.5e-03} & \num{4.6e-03} &
     \num{8.5e-04} & \num{8.8e-04}
     \\
     \rowcolor{green!20}
     Disc decomposition (ours, k=2)  & $4n$ &
     \num{1.6e-10} & \num{1.7e-10} &
     \num{6.5e-04} & \num[math-rm=\mathbf]{6.8e-04} &
     \num{2.4e-03} & \num[math-rm=\mathbf]{2.5e-03} &
     \num{5.5e-04} & \num[math-rm=\mathbf]{5.6e-04} &
     \num{4.5e-07} & \num[math-rm=\mathbf]{4.6e-07}
     \\
  \bottomrule
  \end{tabular}
  }
\end{table*}
\begin{figure*}[tb]
  \centering
  \includegraphics[width=1\linewidth]{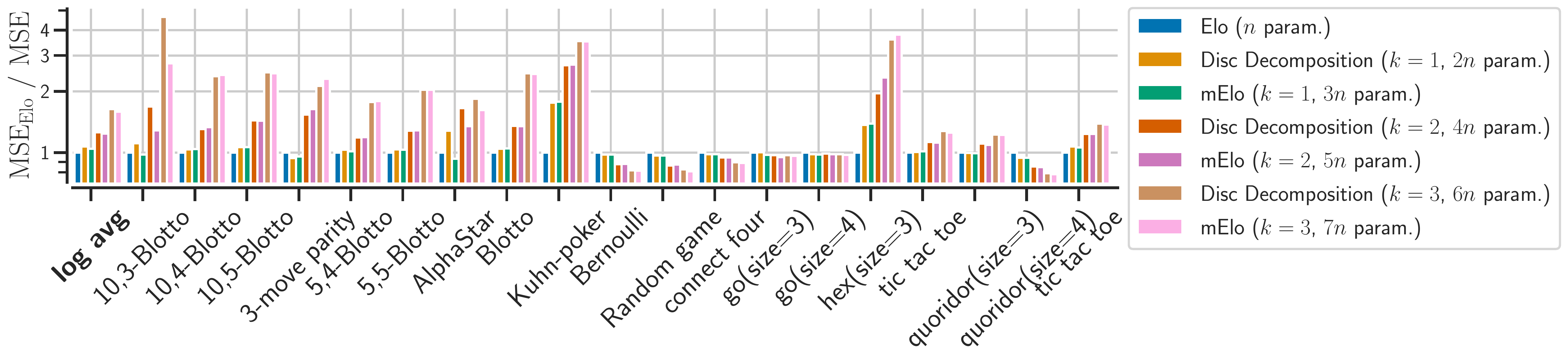}
  \caption{
      \rebutal{\textbf{Prediction Performances (the Higher, the Better) on Unseen Data.}
      The prediction performances on unseen data of the Elo (which has one parameter), the mElo (\citealt{Balduzzi2018}, which has $(2k + 1)n$ parameters), and the proposed disc decomposition (\Cref{alg:otho_decomp}, which has $2kn$ parameters) are compared on a wide range of payoff matrices from \citet{Czarnecki2020}.
      With fewer parameters, the proposed \ours yields better mean squared error than \citet{Balduzzi2018}.}
 }
  \label{fig:pred_spinning_top}
\end{figure*}
%
In this section, we perform experiments
\footnote{Code can be found at \url{https://github.com/QB3/discrating}.}
on real payoff matrices from \citet{Czarnecki2020} and real chess and StarCraft II data.
For each type of data, we propose two kinds of experiments.
First, we qualitatively compare and interpret the first pair of components the proposed \ours ($k=1$ in \Cref{eq:opt_pb_decomposition_gen}), directly in the probability space $\payoffprob{i}{j}$
($\mathcal{L} (x, \hat x) \mapsto \norm{x - \hat x - \1 / 2}^2$ in
\Cref{eq:opt_pb_decomposition_gen}) or in the logit space $\logit \payoffprob{i}{j}$ ($\mathcal{L} : x, \hat x \mapsto \mathrm{bce}(x, \sigma(\hat x))$ in \Cref{eq:opt_pb_decomposition_gen}).
When possible, we reparametrized the \ours $\vect{u}$ and $\vect{v}$ as in \Cref{prop:positive_trans}.
Then some entries of the payoff matrix are hidden (as well as the symmetric entries), and each model is trained and evaluated to predict the missing entries ($20\%$ of the dataset).
For this second experiment, we compare the following models:
\begin{itemize}[noitemsep,itemjoin = \quad,topsep=0pt,parsep=0pt,partopsep=0pt, leftmargin=*]
    \item The usual Elo score \citep{elo1978rating}, which relies on the assumption that the game is purely transitive.
   \item A variant of the Elo score with a quadratic loss (as in Elo ++, \citealt{Sismanis2010}).
    \item \cite{Balduzzi2018} which has $2 k + 1$ parameters ($k \in \{1, 2, 3\}$), where the payoff matrix is decomposed as
    $\mathrm{grad}(\mat{A}) + \mathrm{rot}(\mat{A})$,
    and then
    $\mathrm{rot}(\mat{A})$ is approximated by the Schur decomposition.
    \item \cite{Balduzzi2019}, where the payoff matrix is directly approximated by the Schur decomposition in the probability space (\emph{i.e.}, $2 \mat{P} - 1$), which has $2k$ parameters.
    \item The proposed \ours (\Cref{alg:otho_decomp}) which has $2k$ parameters ($k \in \{1, 2, 3\}$ pairs of components).
\end{itemize}
One can find the mathematical details of each model in \Cref{app:details_comptetitors}.
Note that among the compared methods for matchups probability estimation, only Elo, Elo++, and the proposed \ours (in the transitive case) are ranking methods.
\subsection{Data from \citet{Czarnecki2020}}
%
%
\paragraph{Comments on \Cref{fig:elo_to_disc}}
First, we investigate the visualization and the performance of some payoff matrices from \citet{Czarnecki2020}.
\Cref{fig:elo_to_disc} (and \Cref{fig:transitive_cyclic_app} in \Cref{app:additional_expes}) display the payoff matrices (top row), the representation of \citet{Balduzzi2019} (middle row), and the proposed \ours (bottom row) for games taken from \citet{Czarnecki2020}.
In these representations, one point represents one player: each player $i$ is summarized by two scores $(\playerv{i}, \player{i})$, which correspond to the coordinate of each point.
\Cref{fig:elo_to_disc} displays the representations for multiple payoff matrices: from a pure Elo game (\Cref{ex:elo_game}), to a pure cyclic disc game \Cref{ex:disc_game}, through an average of the payoff matrices in the log space: $\logit(\mat{P})
=
\mathrm{ratio} \cdot \logit(\mat{P}^{\mathrm{Elo}})
+ (1 - \mathrm{ratio}) \cdot \logit(\mat{P}^{\mathrm{Disc}})
$.
For a pure Elo game (top left), one can see that the proposed model recovers perfectly a transitive game (bottom left): the model recovers $\playerv{i}=1$ for all the players $i$.
On the other side of the spectrum, for a disc game (top right), the proposed method can find a disc game representation.
\paragraph{Comments on \Cref{tab:prediction_synthetic}}
It shows the prediction performances (MSE on unseen data) for each of these payoff matrices.
Whereas the Elo score can predict almost perfectly the Elo game, the Elo score fails as soon as the Elo game assumption is violated.
Conversely, the proposed \ours can correctly predict the Elo game and the disc game.
The proposed model can better predict future outcomes, from the Elo game to the disc game.
For the Elo game, Elo, Elo ++ and our method yield similar performances ($\sim 10^{-10}$).
For such small orders of magnitude, we believe that the difference between the methods is mostly due to numerical errors.

\Cref{fig:transitive_cyclic_app} (in \Cref{app:additional_expes}) shows the representations for multiple real games played by computers (machine learning algorithms) from \citet{Czarnecki2020}, including games considered transitive (Go and AlphaStar) and games considered cyclic (Blotto and Kuhn-Poker).
One can see that the convex hull of the proposed representation does not contain $0$ for games considered as mostly transitive: this validates \Cref{prop:disc_game}.
\paragraph{Comments on \Cref{fig:pred_spinning_top,fig_pred_spinning_top_avg}}
\Cref{fig:pred_spinning_top} shows the mean squared error (MSE) between the estimated probabilities and the empirical ones.
More precisely, for each compared method, it shows the MSE of the Elo score divided by the MSE of the method, hence the larger the better.
The predictions from \Cref{fig:pred_spinning_top} show that
the proposed representation yields better predictions for the outcome probabilities.
The improvements of the proposed \ours are more significant for games considered cyclic, such as Blottos or Kuhn-Poker.
`Bernoulli' and `Random Game' are generated at random.
  Since there is only noise in these games, the larger the number of parameters, the larger the overfitting (as observed in \Cref{fig:pred_spinning_top}). We currently do not have explanations for why a larger number of parameters does not yield better predictions on `Connect Four' and `quoridor(size$=4$)' games.
In \Cref{fig_pred_spinning_top_avg}, the averaged MSE across all the games from \citet{Czarnecki2020} shows that, with fewer parameters, the proposed \ours achieves better performances than \citep{Balduzzi2018}.

\begin{figure}[tb]
      \centering
      \includegraphics[width=1\linewidth]{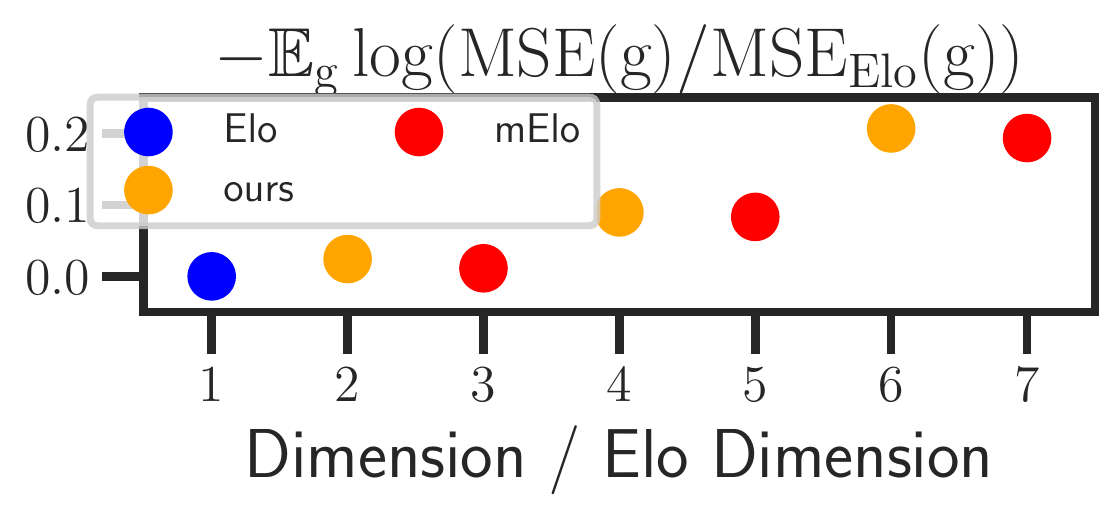}
      \caption{\textbf{Normalized (Log)-Average Prediction Performances on all the Datasets from \citet{Czarnecki2020} (the Higher, the Better).}
      Using fewer parameters, the proposed disc decomposition yields better prediction performances for unseen matchups.
      }
      \label{fig_pred_spinning_top_avg}
  \end{figure}
\begin{figure*}[tb]
      \begin{subfigure}{.5\textwidth}
          \includegraphics[width=1\linewidth]{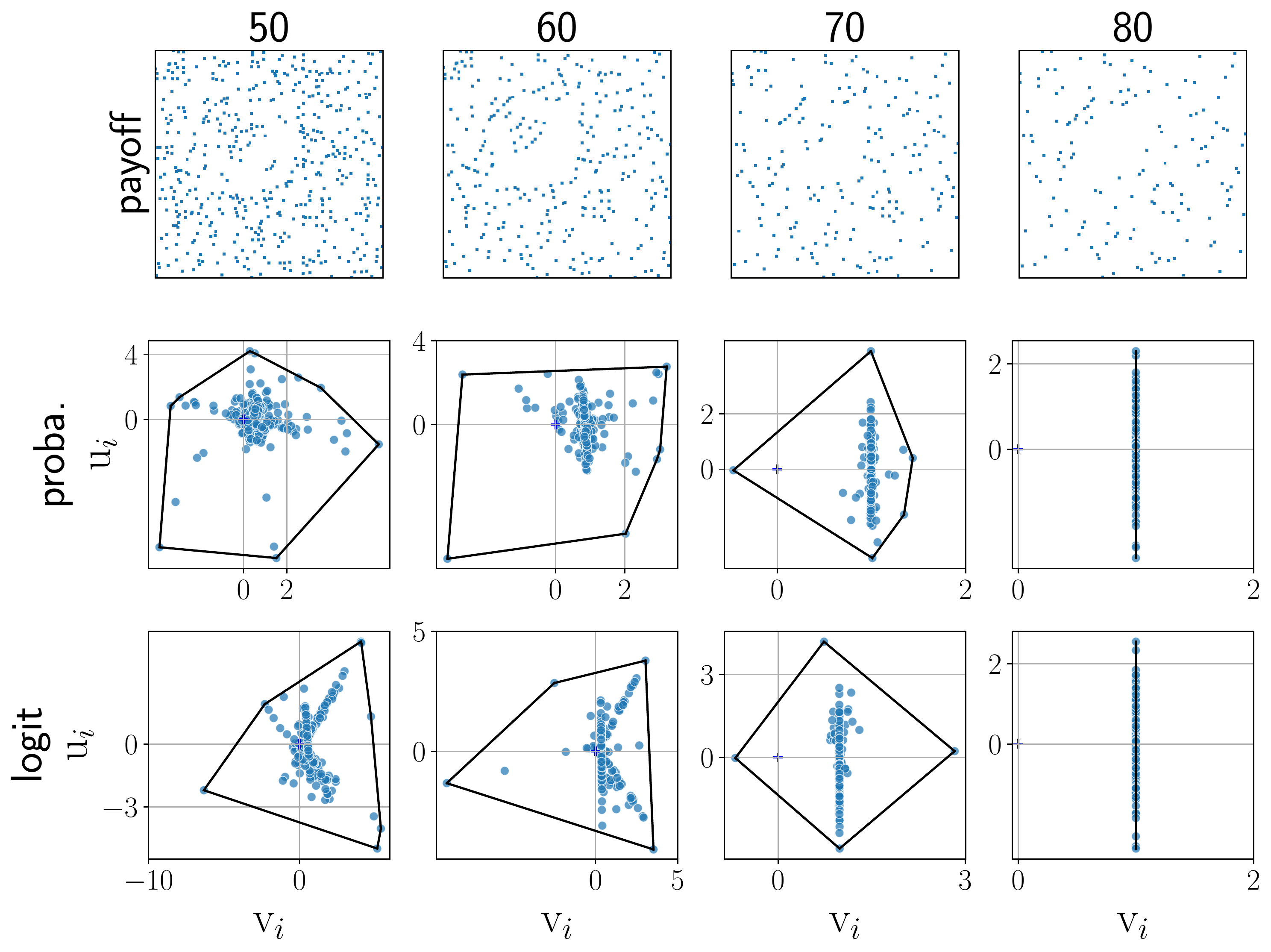}
        \caption{\textbf{Real Lichess  Data.}}
        \label{fig:real_chess_data}
      \end{subfigure}
      \hfill
      \begin{subfigure}{.5\textwidth}
          \includegraphics[width=1\linewidth]{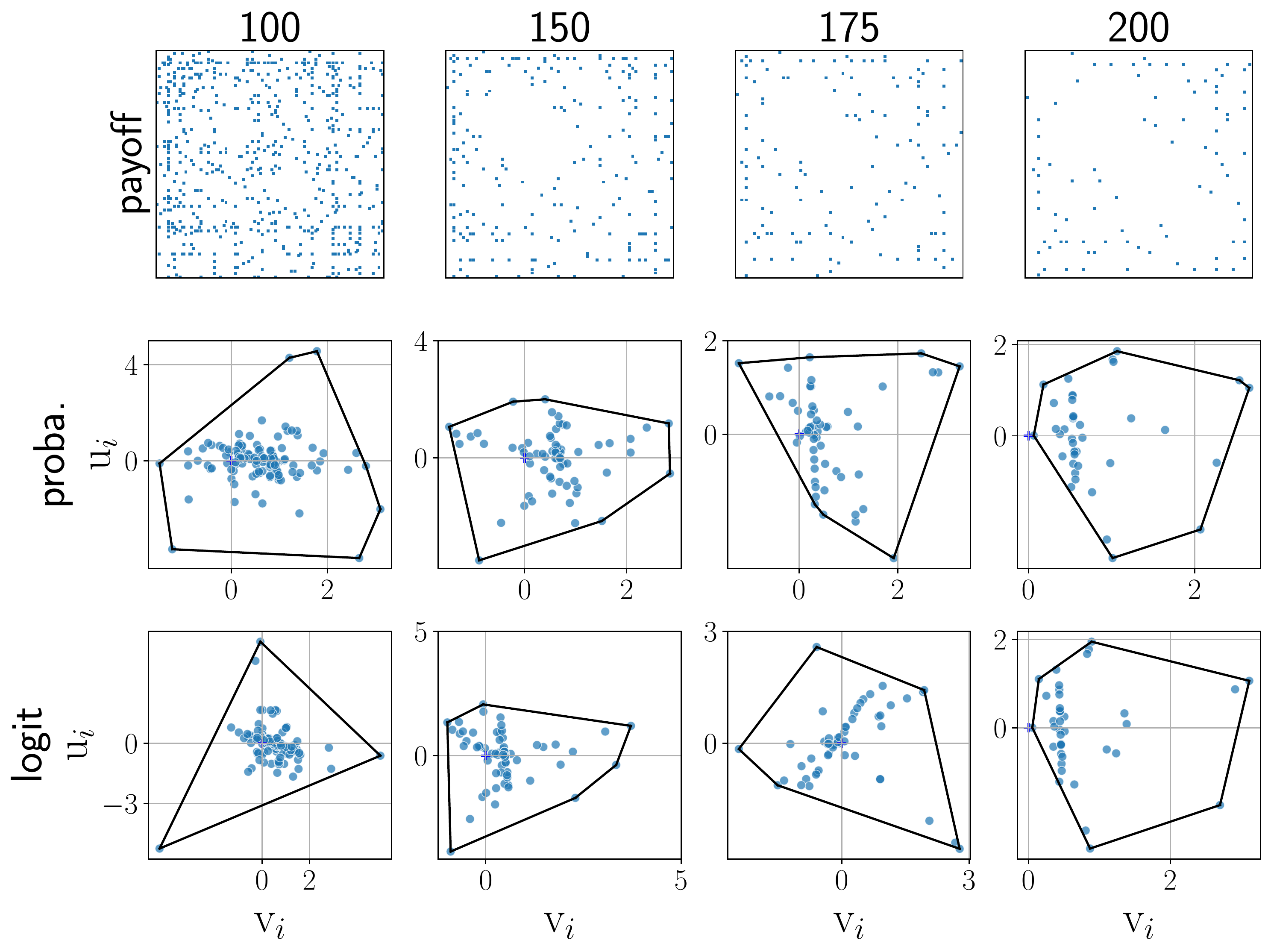}
        \caption{\textbf{Real StarCraft II Data.} }
        \label{fig:real_StarCraft_data}
      \end{subfigure}
      \caption[placeholder]{
      Proposed \ours as a function of the number of matchups for chess and StarCraft II.
      }
      \label{fig:chess_starcraft}
  \end{figure*}

\looseness-1
\subsection{Human Data: Lichess and StarCraft (\Cref{fig:chess_starcraft})}
For the game of chess, we used the
Lichess elite database\footnote{\url{https://lichess.org/team/lichess-elite-database}}.
Lichess is an open-source platform allowing one to play chess online against other players, and the Lichess elite database consists of "all (standard) games from Lichess filtered to only keep games by players rated 2400+ against players rated 2200+".
For our experiments, we used the data from August 2019 to May 2020, which contains more than $4.7$ million games between more than
$\num{40000}$ players.
For StarCraft, we used the
aligulac data \footnote{\url{http://aligulac.com/}}
from tournament with more than $1.7$ millions of games, between $\num{20000}$ players.
In StarCraft, each confrontation consists of multiple "games". We aggregated the games independently of the "race" used by the players. In both cases, the data we are using is public, anonymized, and only concerns non-sensitive information (matchup results).
Elite players usually play online against others more often than amateur ones.
Hence the number of matchups is larger for elite pairs of players, which yields a better estimation of the payoff matrix.
In other words, it makes sense to restrict ourselves to "elite" players.
\paragraph{Influence of the Number of Matchups}
Real-world data has a lot of missing entries; hence the matchups payoff matrices (see \Cref{fig:real_chess_data}) are incomplete.
In \Cref{fig:chess_starcraft} a blue square indicates a confrontation between the players of the corresponding row and column.
To decrease the noise in the estimation of the outcome probability, we only kept the probabilities coming from a large number of matchups between pairs of players.
For chess (\Cref{fig:real_chess_data}) and StarCraft (\Cref{fig:real_StarCraft_data}), we plotted the obtained representations as a function of this number of matchups.
%
\paragraph{Comments on \Cref{fig:chess_starcraft}}
\looseness-1
It displays the proposed \ours computed with \Cref{alg:altmin} in the probability space (middle) and logit space (bottom), for Lichess (\Cref{fig:real_chess_data}) and StarCraft (\Cref{fig:real_StarCraft_data}).
One can see that once the probability matrix is accurate enough (\Cref{fig:real_chess_data}, right, more than $80$ confrontations) the proposed method recovers a fully transitive game ($\playerv{i} > 0$).
The proposed representation also yields a strong transitive component for $60$ and $70$ confrontations.
For StarCraft, even by taking only the entries on the probability matrix with more than $200$ confrontations, one does not recover a representation as transitive as for chess. One interesting conclusion from this experiment is that our score naturally learns from human data that chess and StarCraft are transitive games.
To conclude, as expected from expert knowledge, chess seems to be more transitive than StarCraft (on the studied datasets).
%
\section{CONCLUSION}
In this work, we first recalled some limitations of the Elo score for transitive and cyclic games.
Then we studied in detail the disc game: we showed that depending on ($\vect{u}$, $\vect{v}$), the game is either cyclic or transitive.
Based on this example, we proposed a \ours, which can recover the usual Elo score if the game is an Elo game
and extend the Elo score to transitive non-Elo games and cyclic games.
The theoretical results were extensively validated on real data.

\looseness-1
\paragraph{Limitation}
The main weakness of this work is that it does not take into account
a potential structure in the game.
It is a strength and a weakness. On the one hand, it is a very general method that only requires matchup results between players.
On the other hand, a lot of information could be leveraged from the specificity of the game player, e.g., its (a priori) transitive aspect (as in \citealt{Balduzzi2019}) or its sequential aspect (as in \citealt{herbrich2006trueskill} or \citealt{Sismanis2010}).
A second issue is that the proposed approach is less straightforward to interpret than the Elo. While the Elo assigns only one scalar per player, which can easily yield a ranking interpretation, the proposed scores can be harder to interpret if the game is not transitive.

\paragraph{Societal Impact}
Our work is primarily methodological: we do not see potential negative societal impacts.

\paragraph{Acknowledgments}
QB would like to thank Samsung Electronics Co., Ldt. for funding this research.
GG is supported by an IVADO grant.
\newpage

\clearpage
\bibliography{references}
\bibliographystyle{plainnat}
\newpage
\onecolumn
\appendix

The appendix is organized as follows: \Cref{app:additional_computational_details}
 provides computational details on \Cref{alg:altmin,alg:otho_decomp}.
 \Cref{app:additional_expes} provides additional experiments on data from \citet{Czarnecki2020} and StarCraft data.
 Finally, \Cref{app:missing_proofs} contains the proofs of \Cref{prop:positive_trans,prop:disc_game}, \Cref{thm:at_most_one_transitive,thm:antisym_pca}.
\section{ADDITIONAL COMPUTATIONAL DETAILS}
\label{app:additional_computational_details}

%
\subsection{Computation Details}
\label{app:computational_details}
%
\paragraph{Gradient Computation}
In order to compute the proposed representation one needs to solve the following optimization problem ($k=1$ in \Cref{eq:opt_pb_decomposition_gen})
\begin{align}\label{eq_app:opt_pb_decomposition_gen_gen}
    \argmin_{
        \vect{u}, \vect{v},
        \text{ s.t. } \vect{u}^\top \vect{v} = 0}
    f(\vect{u}, \vect{v})
    \triangleq
    \sum_{i=1}^n
    \sum_{j=1}^n
     \mathcal{L}(\payoffprob{i}{j}, \player{i} \playerv{j} - \playerv{i} \player{j})
    \enspace.
\end{align}
Hence one needs to compute the gradient of $f$ with respect to $\vect{u}$ and $\vect{v}$.
Bellow is provided the formula for these gradients
\begin{align}
    \nabla_{\vect{u}_k} f(\vect{u}, \vect{v})
    &=
    \sum_{i, j}
    \playerv{j}
    \mathcal{L}'(\payoffprob{i}{j}, \player{i} \playerv{j} - \playerv{i} \player{j})
    \mathbb{1}_{i=k}
    - \playerv{i}
    \mathcal{L}'(\payoffprob{i}{j}, \player{i} \playerv{j} - \playerv{i} \player{j})
    \mathbb{1}_{j=k}
    \enspace,
    \\
    &=
    \sum_{j}
    \playerv{j} \mathcal{L}'(\mat{P}_{k, j}, \vect{u}_k \playerv{j} - \vect{v}_k \player{j})
    -
    \sum_{i}
    \playerv{i} \mathcal{L}'(\mat{P}_{i, k}, \player{i} \vect{v}_k - \playerv{i} \vect{u}_k)
    \enspace,
    \\
    \nabla_{\vect{u}} f(\vect{u}, \vect{v})
    &=
    \mathcal{L}'(\mat{P}, \vect{u} \vect{v}^\top - \vect{v} \vect{u}^\top) \vect{v}
    - \mathcal{L}'(\mat{P},  \vect{u} \vect{v}^\top - \vect{v} \vect{u}^\top)^\top \vect{v}
    \enspace.
\end{align}
With similar derivations, one can obtain
\begin{align}
    \nabla_{\vect{v}} f(\vect{u}, \vect{v})
    &=
    \mathcal{L}'(
        \mat{P}, \vect{u} \vect{v}^\top - \vect{v} \vect{u}^\top)^\top \vect{u}
        - \mathcal{L}'(\mat{P}, \vect{u} \vect{v}^\top - \vect{v} \vect{u}^\top) \vect{u}
    \enspace.
\end{align}
For instance, if $\mathcal{L}(\cdot, x) = \norm{\cdot - \sigma(x)}^2 $, then $\mathcal{L}'(\cdot, x) = (1 - \sigma(x)) \sigma(x) (\sigma(x) - \cdot) $.
Once the gradients are computed, one can solve the problem in \ref{eq_app:opt_pb_decomposition_gen_gen} using alternate minimization in $\vect{u}$ and $\vect{v}$:
each optimization subproblem can be solved using the \texttt{scipy} implementation \citep{Virtanen2020} of the l-BFGS algorithm \citep{Liu1989}.
\paragraph{Missing Entries}
Payoff matrices coming from real-world games played by humans usually contain missing entries: one does not have access to the matchups between all the players.
For instance, on the \texttt{Lichess} website, there usually are few confrontations between low-ranked and high-ranked players.
In other words, one only partially has access to $\payoffprob{i}{j}$, for $(i, j)$ in a given set of pairs of players $\mathcal{D}^{\mathrm{obs}} \subset  [n] \times [n]$.
Instead of the problem defined in \Cref{eq_app:opt_pb_decomposition_gen_gen}, the proposed \ours with a partial set of observation $\mathcal{D}^{\mathrm{obs}}$ reads:
\begin{align}\label{eq_app:opt_pb_decomposition_missing_data}
    \argmin_{
        \vect{u}, \vect{v},
        \text{ s.t. } \vect{u}^\top \vect{v} = 0}
    f(\vect{u}, \vect{v})
    \triangleq
    \sum_{(i, j) \in \mathcal{D}^{\mathrm{obs}}}
     \mathcal{L}(\payoffprob{i}{j}, \player{i} \playerv{j} - \playerv{i} \player{j})
    \enspace.
\end{align}
%
\subsection{Detailed on the Compared Methods}
\label{app:details_comptetitors}
%
\Cref{tab:prediction_synthetic} compares the performance on unseen data of the following models for payoff matrices from \citet{Czarnecki2020}
\begin{itemize}
    \item The usual Elo score \citep{elo1978rating}
    \begin{align}
        \label{eq_app:elo_stationary}
        \vect{u}^{\mathrm {Elo}}
        \in
        \argmin_{\vect{u} } \sum_{i, j}
        \mathcal{L}
        (\payoffprob{i}{j}, \sigma(\player{i} - \player{j})),
        \, \text{ with }
        \mathcal{L} : (x, \hat x) \mapsto \mathrm{bce}(x, \hat x)
        \enspace.
    \end{align}
    \item Elo ++ \citep{Sismanis2010}, without the sequential aspect
    \begin{align}
        \label{eq_app:elopp_stationary}
        \vect{u}^{\mathrm {Elo ++}}
        \in
        \argmin_{\vect{u} } \sum_{i, j}
        \mathcal{L}
        (\payoffprob{i}{j}, \sigma(\player{i} - \player{j})),
        \, \text{ with }
        \mathcal{L} :(x, \hat x)
        \mapsto
        \frac{1}{2}
        \normin{x - \hat x}^2
        \enspace.
    \end{align}
    \item This work (logit space)
    \begin{align}\label{eq_app:opt_pb_decomposition_gen}
        (\vect{u}^{\mathrm{Disc}}, \vect{v}^{\mathrm{Disc}} )
        \in
        \argmin_{
            \vect{u}, \vect{v},
            \text{ s.t. } \vect{u}^\top \vect{v} = 0}
        \sum_{i, j}
         \mathcal{L}(\payoffprob{i}{j}, \player{i} \playerv{j} - \playerv{i} \player{j}) ,
         \, \text{ with }
         \mathcal{L} : (x, \hat x) \mapsto \mathrm{bce}(x, \sigma(\hat x))
        \enspace.
    \end{align}
    \item This work (probability space), which can be seen as a representation proposed in \citet{Balduzzi2019}
    \begin{align}\label{eq_app:opt_pb_decomposition_gen_proba}
        (\vect{u}^{\mathrm{Disc}}, \vect{v}^{\mathrm{Disc}} )
        \in
        \argmin_{
            \vect{u}, \vect{v},
            \text{ s.t. } \vect{u}^\top \vect{v} = 0}
        \sum_{i, j}
         \mathcal{L}(\payoffprob{i}{j}, \player{i} \playerv{j} - \playerv{i} \player{j}) ,
         \, \text{ with }
         \mathcal{L} :(x, \hat x)
         \mapsto
         \frac{1}{2}
         \normin{x - \frac{1}{2} - \hat x}^2
        \enspace.
    \end{align}
    \item m-Elo \citep{Balduzzi2018}, with
    $\mat{\bar P} \triangleq (\frac{1}{n} \sum_j \payoffprob{i}{j})_{1 \leq i \leq n}$
    \begin{align}\label{eq_app:opt_melo}
        (\vect{u}^{\mathrm{m-Elo}}, \vect{v}^{\mathrm{m-Elo}} )
        \in
        \argmin_{
            \vect{u}, \vect{v},
            \text{ s.t. } \vect{u}^\top \vect{v} = 0}
        \sum_{i, j}
         \mathcal{L}(\payoffprob{i}{j} - (\mat{\bar P}_{i} - \mat{\bar P}_{j}), \player{i} \playerv{j} - \playerv{i} \player{j})
        \enspace,
    \end{align}
     with
    $\mathcal{L} :(x, \hat x)
    \mapsto
    \frac{1}{2}
    \normin{x - \frac{1}{2} - \hat x}^2$.
\end{itemize}
This work and the m-Elo are computed using \Cref{alg:otho_decomp}.
Each optimization problem in the alternate minimization (\Cref{alg:altmin}) is solved using the \texttt{scipy} \citep{Virtanen2020} implementation of l-BFGS \citep{Liu1989}.
%
\section{ADDITIONAL EXPERIMENTS}
\label{app:additional_expes}
\subsection{Data from \citet{Czarnecki2020}}
\Cref{fig:elo_to_disc_app,fig:transitive_cyclic_app} are similar to \Cref{fig:elo_to_disc}, but the Nash clustering visualization from \citet{Czarnecki2020} has been added.
For each payoff matrix, \citet{Czarnecki2020} successively compute Nash equilibria of empirical games to cluster players into level sets of ``strength''.
For each dot $i$, $\playerv{i}$ corresponds to the clustered index and $\player{i}$ to the fraction of the population beaten by the cluster (this is different from our representation where each dot corresponds to a player).
The number of clusters can be interpreted as a measure of transitivity: the larger the number of clusters, the more transitive the game.
\begin{figure}[tb]
    \centering
    \includegraphics[width=\linewidth]{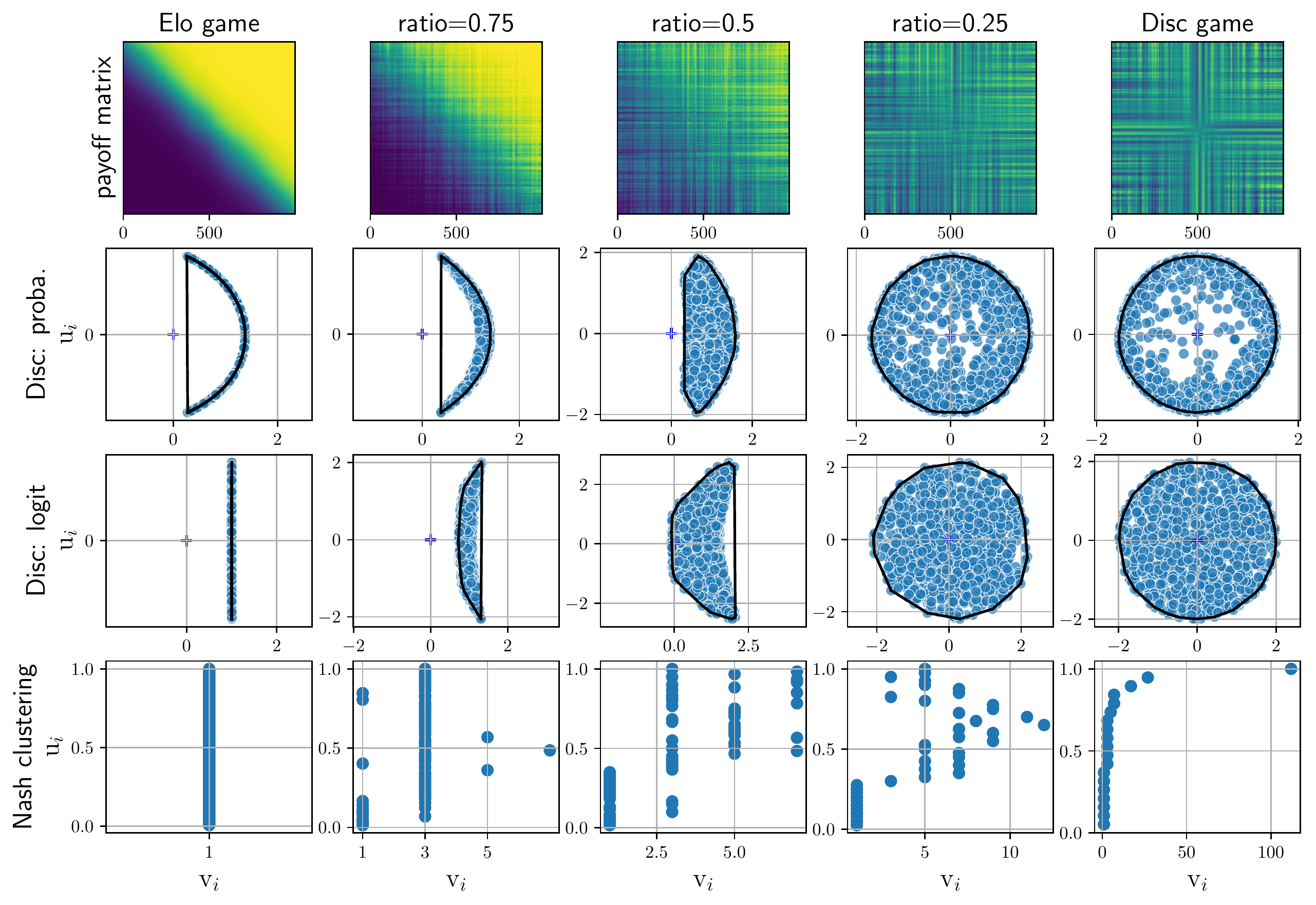}
    \caption{
        \textbf{From Elo to Disc Game.}
    Payoff matrices (top) and visualization using Algorithm~\ref{alg:altmin} in the probability space (second row), logit space (third row), or \citet{Czarnecki2020} (bottom), for multiple games with payoff
    $\logit(\mat{P})
    =
    \mathrm{ratio} \cdot \logit(\mat{P}^{\mathrm{Elo}})
    + (1 - \mathrm{ratio}) \cdot \logit(\mat{P}^{\mathrm{Disc}})
    $.}
    \label{fig:elo_to_disc_app}
\end{figure}
\begin{figure}[tb]
    \centering
    \includegraphics[width=\linewidth]{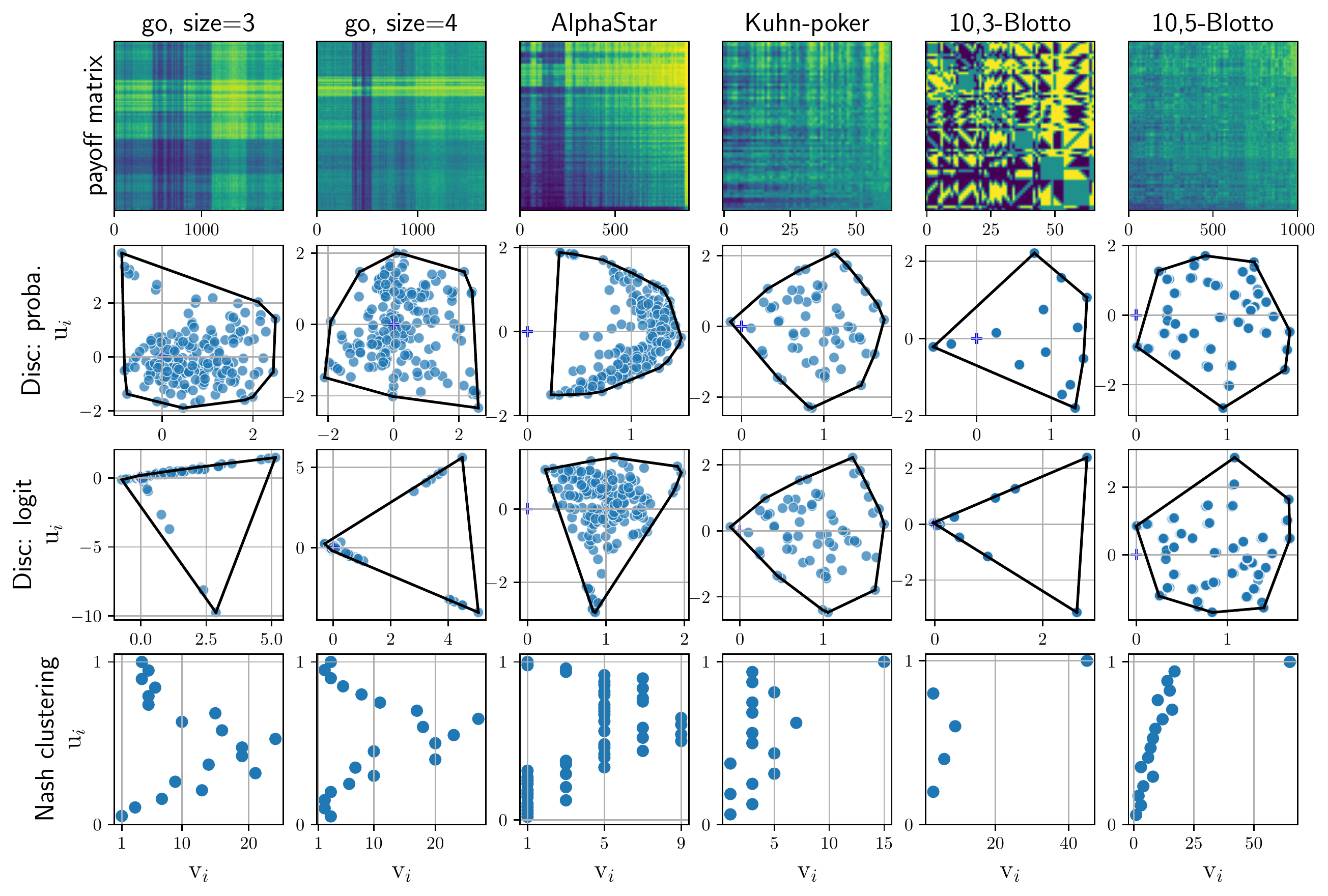}
    \caption{\textbf{Transitive and Cyclic Games.}
    Payoff matrices (top) and visualization using Algorithm~\ref{alg:altmin} in the probability space (second row), logit space (third row), or \citet{Czarnecki2020} (bottom), for multiple games with payoff
    from \citet{Czarnecki2020}.}
    \label{fig:transitive_cyclic_app}
\end{figure}
%
%
\subsection{Prediction Performance for StarCraft data}
%
\Cref{fig:pred_starcraft} shows the prediction performance on unseen entries for StarCraft II confrontations.
As observed in \citet{Sismanis2010}, adding an extra parameter makes the models significantly overfit.
That is why a regularization parameter was added to prevent overfitting:
\begin{align}\label{eq_app:opt_pb_decomposition_gen_reg}
    (\vect{u}^{\mathrm{Disc}}, \vect{v}^{\mathrm{Disc}} )
    \in
    \argmin_{
        \vect{u}, \vect{v},
        \text{ s.t. } \vect{u}^\top \vect{v} = 0}
    \sum_{i, j}
     \mathcal{L}(\payoffprob{i}{j}, \player{i} \playerv{j} - \playerv{i} \player{j})
    + \frac{\lambda}{2} \sum_i (\playerv{i} - 1)^2
     \enspace,
\end{align}
where the regularization parameter $\lambda$ is chosen using cross-validation.

One can see on \Cref{fig:pred_starcraft} that the proposed regularized disc rating can yield better predictions than the usual Elo score.
But improvements are not as impressive as for data from \citealt{Czarnecki2020} (\Cref{fig:pred_spinning_top}).
This might be due to a more difficult estimation of the payoff matrix: one only has access to an incomplete and potentially noisy estimation of the payoff matrix.
\begin{figure}[tb]
    \centering
    \includegraphics[width=1\linewidth]{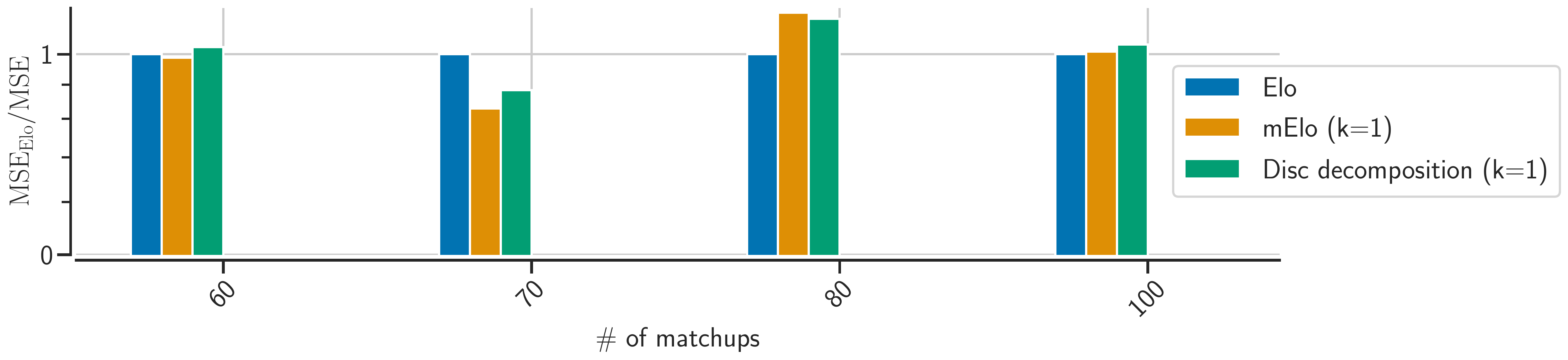}
    \caption{Prediction performances as a function of the number of matchups for the StarCraft data.}
    \label{fig:pred_starcraft}
\end{figure}
%

\section{PROOFS OF THEOREMS AND PROPOSITIONS}
\label{app:missing_proofs}

\end{document}